\documentclass[aps,prb,reprint,superscriptaddress,longbibliography]{revtex4-2}

\usepackage{graphicx}
\graphicspath{{./}{./figure/}}

\usepackage{amsmath,amssymb,bm,braket}
\usepackage{enumitem}
\usepackage{comment}
\usepackage{multirow}
\usepackage{tabularx}
\newcolumntype{Y}{&gt;{\centering\arraybackslash}X} %

\usepackage{xcolor}
\newcommand{\change}[1]{\textcolor{blue}{#1}}
\renewcommand{\change}[1]{#1}

\usepackage{diagbox}

\usepackage{amsthm}

\newtheorem{thm}{Theorem}
\newtheorem{prop}[thm]{Proposition}

\newcommand{\eq}[1]{\begin{align} #1 \end{align}}
\newcommand{\nx}{ \nonumber \\}
\newcommand{\bmat}[1]{\begin{matrix} #1 \end{matrix}}
\newcommand{\bmatc}[2]{\begin{array}{#2} #1 \end{array}}
\newcommand{\sumi}{\sum_{i=1}^N}

\newcommand{\boldA}{\boldsymbol{A}}
\newcommand{\boldB}{\boldsymbol{B}}
\newcommand{\boldC}{\boldsymbol{C}}

\newcommand{\ak}{{\boldA^k}}

\newcommand{\aki}{{\boldA^k_i}}
\newcommand{\akmi}{{\boldA^{k-1}_i}}
\newcommand{\bkp}{{\boldB^{k+1}}}

\newcommand{\bkpi}{{\boldB^{k+1}_i}}
\newcommand{\bki}{{\boldB^k_i}}

\newcommand{\Ak}{{\boldA^3}}

\newcommand{\Aki}{{\boldA^3_i}}

\newcommand{\Bkp}{{\boldB^4}}
\newcommand{\Bk}{{\boldB^3}}
\newcommand{\Bkpi}{{\boldB^4_i}}
\newcommand{\Bki}{{\boldB^3_i}}

\newcommand{\ep}{{E_{+1}}}
\newcommand{\ez}{{E_0}}
\newcommand{\fpp}{{F_{+2}}}
\newcommand{\fp}{{F_{+1}}}
\newcommand{\fz}{{F_0}}
\newcommand{\fm}{{F_{-1}}}
\newcommand{\fmm}{{F_{-2}}}
\newcommand{\cep}{e_{+1}}
\newcommand{\cez}{e_0}
\newcommand{\cem}{e_{-1}}
\newcommand{\cfpp}{f_{+2}}
\newcommand{\cfp}{f_{+1}}
\newcommand{\cfz}{f_{0}}
\newcommand{\cfm}{f_{-1}}
\newcommand{\cfmm}{f_{-2}}

\newcommand{\dg}{^\dagger}
\newcommand{\Aone}{A_{(1)}}
\newcommand{\Atwo}{A_{(2)}}
\newcommand{\Athree}{A_{(3)}}
\newcommand{\Afour}{A_{(4)}}
\newcommand{\Bone}{B_{(1)}}
\newcommand{\Btwo}{B_{(2)}}
\newcommand{\Bthree}{B_{(3)}}
\newcommand{\Bfour}{B_{(4)}}
\newcommand{\Cone}{C_{(1)}}
\newcommand{\Ctwo}{C_{(2)}}
\newcommand{\Cthree}{C_{(3)}}
\newcommand{\Cfour}{C_{(4)}}
\newcommand{\Ckmm}{C_{(k-2)}}
\newcommand{\Ckm}{C_{(k-1)}}

\newcommand{\Ic}[1]{\multicolumn{1}{|c}{#1}}
\newcommand{\cI}[1]{\multicolumn{1}{c|}{#1}}
\newcommand{\ezo}[1]{E_0^{\otimes #1}}

\newcommand{\im}{\mathrm{i}}

\usepackage{hyperref}
\hypersetup{colorlinks=true, citecolor=magenta, linkcolor=blue, urlcolor=cyan}

\begin{document}

\title{Absence of \change{nontrivial} local conserved quantities in the spin-1 bilinear-biquadratic chain and its anisotropic extensions}

\author{Akihiro Hokkyo}
\email{hokkyo@cat.phys.s.u-tokyo.ac.jp}
\affiliation{Department of Physics, Graduate School of Science, The University of Tokyo, 7-3-1 Hongo, Bunkyo, Tokyo 113-8654, Japan}

\author{Mizuki Yamaguchi}
\email{yamaguchi-q@g.ecc.u-tokyo.ac.jp}
\affiliation{Department of Basic Science, Department of Multidisciplinary Sciences, Graduate School of Arts and Sciences, The University of Tokyo, 3-8-1 Komaba, Meguro, Tokyo 153-8902, Japan}

\author{Yuuya Chiba}
\email{yuya.chiba@riken.jp}
\affiliation{ Nonequilibrium Quantum Statistical Mechanics RIKEN Hakubi Research Team, \change{Pioneering Research Institute (PRI)}, RIKEN, 2-1 Hirosawa, Wako, Saitama 351-0198, Japan}

\begin{abstract}
We provide a complete classification of the
integrability and nonintegrability of the spin-1 bilinear-biquadratic model with a uniaxial anisotropic field, 
which includes the Heisenberg model and the Affleck-Kennedy-Lieb-Tasaki model. 
\change{It is rigorously shown that, within this class, 
the only integrable systems are those 
that have been solved by the Bethe ansatz method%
, 
and that all other systems are nonintegrable, in the sense that
they do not have nontrivial local conserved quantities. 
Here, ``nontrivial'' excludes quantities like the Hamiltonian or the total magnetization, 
and ``local'' refers to sums of operators that act only on sites within a finite distance.
This result establishes the nonintegrability of the Affleck-Kennedy-Lieb-Tasaki model and, consequently, 
demonstrates that the quantum many-body scars observed in this model emerge independently of any conservation laws of local quantities.} 
Furthermore, 
\change{we extend the proof of nonintegrability to 
more general spin-1 models that encompass anisotropic extensions of the bilinear–biquadratic Hamiltonian,
and completely classified the integrability of generic Hamiltonians that possess translational symmetry, $U(1)$ symmetry, time-reversal symmetry, and spin-flip symmetry.}
Our result has accomplished a breakthrough in nonintegrability proofs by expanding their scope to spin-1 systems.
\end{abstract}

\maketitle

\section{Introduction}

Understanding equilibrium and nonequilibrium physics in quantum many-body systems is one of the most challenging and significant problems. 
Quantum integrable systems, 
whose energy spectrum and eigenstates can be obtained by using analytical methods such as the Bethe ansatz~\cite{Bethe1931},
have played a pivotal role in such studies
because they provide exact solutions for various quantities.
Their solvability is considered to be closely related to the existence of an infinite number of local conserved quantities.
For instance, in the algebraic Bethe ansatz~\cite{Sklyanin1979,Takhtadzhan1979,Baxter1982,Korepin1993,Takahashi1999}, 
the monodromy matrix used to generate the eigenstates also serves as a generating function for the local conserved quantities.

Despite their great contributions, some phenomena cannot be understood through integrable systems.
This is because the existence of an infinite number of local conserved quantities is incompatible with many empirical laws of macroscopic systems. 
A prominent example is thermalization. 
It has recently been recognized that many isolated quantum many-body systems tend 
to relax toward thermal equilibrium~\cite{DAlessio2016,Mori2018}. 
Such relaxation is usually explained by the eigenstate thermalization hypothesis~\cite{Deutsch1991,Srednicki1994,Rigol2008}, 
which states that every energy eigenstate is thermal, i.e., represents thermal equilibrium. 
In contrast, integrable systems require an infinite number of parameters to describe their stationary states, 
which are characterized by the generalized Gibbs ensemble~\cite{Rigol2007,Vidmar2016} 
and thus lie beyond the framework of conventional thermodynamics.
Furthermore, integrability can lead to deviations from empirical laws even at the level of linear response.
The Green-Kubo formula~\cite{Green1954,Kubo1957} may yield results 
that disagree with empirical macroscopic laws when applied to integrable systems.
The Mazur-Suzuki bound~\cite{Mazur1969,Suzuki1971,Sirker2020} clarifies this issue:
For the Green-Kubo formula to yield results consistent with empirical behavior, 
no local conserved quantity should overlap with the relevant observables, such as the magnetization (in the case of susceptibilities)~\cite{Chiba2020,Chiba2023}
or the current operators (in the case of transport coefficients)~\cite{Saito2003,Sirker2020}. 
These examples indicate that, 
in order to understand the empirical behavior of macroscopic systems, 
it is necessary to rule out the possibility of integrability, 
as integrable systems tend to exhibit anomalous behavior.

\change{Recent studies have revealed some anomalous behaviors regarding thermalization even in systems without any known integrable structure.
Major examples include the quantum many-body scar (QMBS) state~\cite{Bernien2017,Turner2018,Turner2018a}, 
which is a nonthermal energy eigenstate that appears in systems not known to be integrable.
The QMBS provides a violation of thermalization 
with an origin that appears distinct from integrability, 
indicating that the problem of thermalization has a richer structure. 
While numerical evidence strongly suggests that such systems do not possess conventional integrable structures~\cite{Turner2018,Moudgalya2018}, 
the precise relationship between the QMBS and integrability remains unclear. 
In fact, models with the QMBS are sometimes viewed as deformations of integrable systems~\cite{Turner2018a,Mark2020,Moudgalya2020,ODea2020}, 
raising the possibility that 
some hidden local conserved quantities may still play a role. 
Therefore, whether the QMBS is truly independent of 
local conservation laws
remains an open and fundamental question.} 

One of the most widely used models in the QMBS studies is the spin-1 Affleck-Kennedy-Lieb-Tasaki (AKLT) model~\cite{Affleck1987,Affleck1988}.
Nonthermal energy eigenstates in this model have been constructed exactly~\cite{Moudgalya2018,Moudgalya2018a,Mark2020},
and they are found to be equally spaced throughout the energy spectrum.
These states lead to perfect revivals in observables and fidelity for specific initial states~\cite{Mark2020}. 
Since such nonthermal eigenstates can be written down analytically
unlike typical thermal eigenstates, 
this model contributes to an exact understanding of the problem of thermalization. 
Nonetheless,
the nature of conserved quantities and the relation to integrability in this model are still not fully understood.

A broader perspective on the integrability can be gained by studying the spin-1 BLBQ model,
which interpolates between the AKLT model, the Heisenberg model~\cite{White1993}, 
and several integrable points~\cite{Uimin1970,Lai1974,Sutherland1975,Schmitt1996,Takhatajan1982,Babujian1982,Parkinson1987,Klumper1989,Barber1989}.
The BLBQ model exhibits various intriguing phenomena, such as the quantum phase transition~\cite{Mikeska2004,Rodriguez2011,Chiara2011,Quella2021,Rakov2022}, the Kibble-Zurek physics~\cite{Dhar2022}, 
and the emergence of quantum chaos~\cite{Santos2020}.
In particular, transport phenomena have been widely explored across different parameter regimes, both near and away from known integrable points~\cite{Ilievski2014,Prosen2015,Piroli2016,Dupont2020,Lima2023}.
When anisotropy is introduced, the transport behavior becomes even more diverse~\cite{Dupont2020}, showing dramatic changes among diffusive, ballistic, and superdiffusive regimes. 
In light of these backgrounds, a complete classification of whether the BLBQ model and its anisotropic extensions are integrable deserves further study to grasp the whole phase diagram of such phenomena.

Recently, 
a rigorous method for exhaustively identifying local conserved quantities was introduced by Shiraishi~\cite{Shiraishi2019}. 
This approach has been successfully applied to various quantum systems~\cite{Chiba2024,Park2025graph,Shiraishi2024,Yamaguchi2024},
\change{revealing that many models are \textit{nonintegrable} in a strict sense:
They possess no local conserved quantities beyond linear combinations of the Hamiltonian and on-site operators, such as the total magnetization.}
However, no such rigorous result has yet been obtained for the BLBQ model, 
despite its fundamental importance. 
In fact, no such result is available for any spin-1 system due to the complicated algebraic structure of the spin-1 operators compared to the spin-$1/2$ case, as pointed out in the pioneering work~\cite{Shiraishi2019}.

In this paper, 
we overcome these difficulties
and give a complete classification of the integrability and the nonintegrability of the spin-1 BLBQ model with a uniaxial anisotropic field. 
We have rigorously shown that all systems in this model are nonintegrable 
in the above sense
except for 
\change{those 
that have been solved by the Bethe ansatz method}%
~\cite{Parkinson1987,Klumper1989,Barber1989, Uimin1970,Lai1974,Sutherland1975,Schmitt1996,Takhatajan1982,Babujian1982}.
In particular, the AKLT model is shown to be nonintegrable, 
which indicates the independence of the QMBS and integrability. 
\change{We also completely classified the integrability of generic Hamiltonians possessing translational, $U(1)$, 
time-reversal, and spin-flip symmetries, including those with anisotropic interactions.}

The paper is organized as follows.
In Sec.~\ref{sec:setup}, 
we state the main results along with a rigorous definition of local conserved quantities. 
Before proceeding to the proof, 
we summarize the overview in Sec.~\ref{sec:strategy}, 
including the properties of the special operator basis that
we utilize in the next section.
In Sec.~\ref{sec:proof}, we give proofs of nonintegrability, 
except for two special zero-parameter cases and 
the case where nonintegrability appears thanks to the uniaxial anisotropy.
The proof of the former cases is provided in Appendix~\ref{sec:singular_case}, 
and that of the latter case is explained in Sec.~\ref{sec:necessary}. 
We also discuss a more general model including the BLBQ model in Sec.~\ref{sec:necessary}.
We conclude with a brief summary and outlook in Sec.~\ref{sec:conclusion}.

\section{Setup and Main Result}\label{sec:setup}

We consider the spin-1 bilinear-biquadratic chain with an anisotropic field on $N$ sites with the periodic boundary condition.
Hereafter, this model is simply referred to as the BLBQ model.
The BLBQ model is described by the following Hamiltonian:
\begin{align}
    H=\sumi \left(J_{1}\vec{S}_{i}\cdot\vec{S}_{i+1}+J_{2}(\vec{S}_{i}\cdot\vec{S}_{i+1})^2+D (S_i^z)^2 \right)~.
    \label{eq:H_BLBQ}
\end{align}
The first term represents the bilinear interaction, i.e., 
the Heisenberg interaction.
The second term describes the biquadratic interaction, which is specific to spin $S\geq 1$ systems.
The last term breaks $SU(2)$-symmetry, 
corresponding to single-spin uniaxial anisotropic field or quadratic Zeeman effect, which is also specific to $S \geq 1$.

We shall classify the above model into integrable and nonintegrable systems. 
\change{Here, we characterize \textit{nonintegrability} by the absence of nontrivial local conserved quantities (defined below) and 
{\it integrability} by possessing $\Omega(N)$ number of local conserved quantities~\cite{Caux2011,Mori2018,Shiraishi2019,Yamaguchi2024F}. 
Here, $\Omega(N)$ indicates that the quantity grows proportionally to $N$ or faster.}
In order to state our claim in a rigorous manner, 
we define below the precise meaning of local conserved quantity.
First, an operator is called a {\it $k$-local operator} if it acts on $k$ consecutive sites.
For instance, $S^x_{i} S^y_{i+1} S^z_{i+2}$ is a 3-local operator and $S^z_{i} S^x_{i+4}$ is a 5-local operator.
Next, an operator is called a {\it $k$-local quantity} if it can be written by a sum of $l$-local operators with $l\leq k$ and cannot be expressed by a sum of those with $l\leq k-1$.
Note that a $k$-local quantity is not necessarily a shift sum of $k$-local operators, e.g., both $\sumi S^z_{i}$ and $\sum (-1)^i S^z_{i}$ are $1$-local quantities.
A {\itshape $k$-local conserved quantity} is defined as a $k$-local quantity 
that commutes with the Hamiltonian.
We refer to $k$-local conserved quantities with 
$k\leq N/2$~\footnote{
We call $k$-local conserved quantities ``local'' ones even when $k=N/2$ 
because the corresponding $k$-local quantities in much larger systems also become conserved quantities, 
while it is not the case when $k>N/2$. 
For instance, 
we consider a spin chain with $N$ sites and nearest-neighbor interactions, 
and define a $3$-local quantity $Q^{[N]}=\sum_{i=1}^{N}S^x_{i} S^y_{i+1} S^z_{i+2}$. 
If it is the $3$-local conserved quantity for $N=6$,
then it is also a conserved quantity of a larger system $N>6$. 
On the other hand, 
although $H^2$ is an $(N_1/2+2)$-local conserved quantity in the nearest-neighbor-interacting systems on an $N=N_1$ chain, 
the corresponding $(N_1/2+2)$-local quantity 
(which is different from the square of the Hamiltonian for the large system)
is not conserved when $N>N_1$ in general. 
See also Sec.~VI~A in Ref.~\cite{Chiba2024}.
} as {\it local conserved quantities} and those with $k\geq 3$ as {\it nontrivial} ones.

The BLBQ model is known to be integrable~\cite{Parkinson1987,Klumper1989,Barber1989, Uimin1970,Lai1974,Sutherland1975,Schmitt1996,Takhatajan1982,Babujian1982}, i.e., possess an $\Omega(N)$ number of local conserved quantities, given any of the following three parameter sets:
\eq{
J_1 &= 0 \quad {\rm and}\quad D = 0~, \nx
J_1 &= J_2~, \nx
J_1 &= -J_2 \quad {\rm and}\quad  D = 0~.
\label{eq:except}
}
In contrast, we claim that the BLBQ model~\eqref{eq:H_BLBQ} is generically nonintegrable as long as none of the conditions~\eqref{eq:except} is satisfied. 
\begin{figure}
    \centering
    \includegraphics[width=1\linewidth]{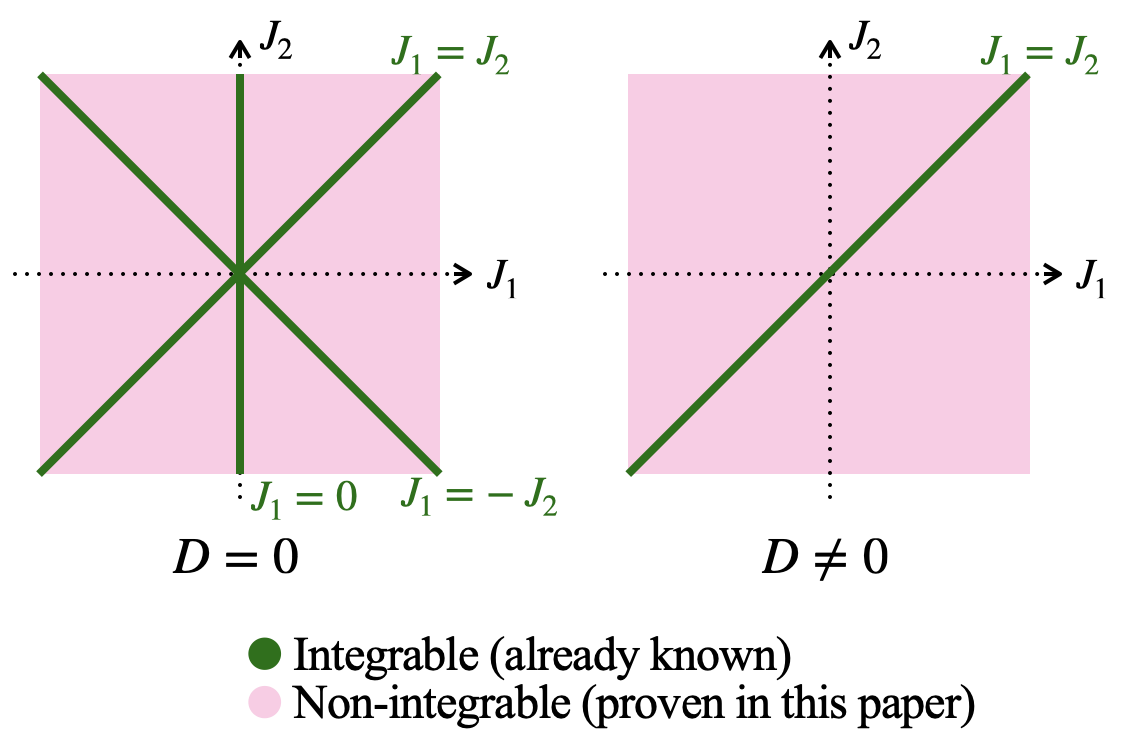}
    \caption{\label{fig:1}
    We perform a complete classification of the BLBQ model into integrable and nonintegrable systems.
    In the case without anisotropic field $D$, shown on the left, there are three known integrable systems, as indicated by three lines. 
    In the case with nonzero $D$, shown on the right, there is one known integrable system.
    We prove the nonintegrability of the complement of these integrable systems.
    }
\end{figure}
\begin{thm}[Main Result]\label{thm:main}
In the BLBQ model~\eqref{eq:H_BLBQ} not satisfying Eq.~\eqref{eq:except},  $k$-local conserved quantities with $3 \leq k \leq N/2$ are absent.
\end{thm}
This theorem provides a definitive answer that the BLBQ model~\eqref{eq:H_BLBQ} contains no further integrable systems (see Fig.~\ref{fig:1}). 
In other words, 
the systems possessing $\Omega(N)$ number of local conserved quantities satisfies one of the conditions~\eqref{eq:except}, 
which are solved by using the Bethe ansatz method~\cite{Parkinson1987,Klumper1989,Barber1989, Uimin1970,Lai1974,Sutherland1975,Schmitt1996,Takhatajan1982,Babujian1982}.
Our result rigorously shows that, in the BLBQ model, 
the set of integrable systems coincides with that of Yang--Baxter-integrable systems, 
which are classified in Ref.~\cite{kennedy1992solutions}.

Moreover, Theorem~\ref{thm:main} also manifests that there are no {\it partially integrable systems}, which possess finite number of nontrivial local conserved quantities.
That is, the BLBQ models are clearly classified into integrable systems with $\Omega(N)$ local conserved quantities and nonintegrable systems with no nontrivial ones.

As special cases of our results, the AKLT model ($J_2 = J_1/3$) and spin-1 Heisenberg chain ($J_2=0$) are proved to be nonintegrable.
This rigorously demonstrates that anomalous behaviors regarding thermalization in the AKLT model~\cite{Moudgalya2018,Moudgalya2018a} cannot be explained from integrability, as is widely expected from numerical evidence such as investigation~\cite{Moudgalya2018} of the level spacing statistics~\cite{Atas2013,Mehta2004}. 
Note that the nonintegrability of the PXP model, which has another type of QMBS, has been shown in Ref.~\cite{Park2025graph}.

We add to Theorem~\ref{thm:main} the analysis of $1$-local and $2$-local conserved quantities and obtain the following theorem.

\begin{thm}\label{thm:1and2}
In the BLBQ model~\eqref{eq:H_BLBQ} not satisfying Eq.~\eqref{eq:except}, $k$-local conserved quantities with $k\leq N/2$ are restricted to linear combinations of the following:
\begin{enumerate}[label={(\roman*)},ref={\roman*}]
    \item its own Hamiltonian: $H$;
    \item the total magnetization in the $z$ direction: $M_z = \sumi S_i^z$;
    \item the total magnetization in the $x$ and $y$ directions: $M_x$ and $M_y$, if $D=0$ holds; and
    \item the staggered quadratic spins:
\eq{
&\sumi (-1)^i(S_i^z)^2~,\nx 
&\sumi (-1)^i((S_i^x)^2-(S_i^y)^2)~,\nx 
&\sumi (-1)^i(S_i^x S_i^y+S_i^y S_i^x)~,
}
if $J_1=0$ holds and $N$ is even.
\end{enumerate}
\end{thm}

\section{Proof strategy}\label{sec:strategy}
{%
\renewcommand{\em}{{E_{-1}}}
\subsection{Basis}
Before going to the outline of the proof, we introduce a basis of operators, which play a crucial role in our proof of nonintegrability.

Previous studies on nonintegrability proofs~\cite{Shiraishi2019,Chiba2024,Park2025graph,Shiraishi2024,Yamaguchi2024}, which deal with spin-1/2 chains, exploited Pauli matrices $\{ \sigma^x, \sigma^y, \sigma^z, I \}$ as a basis of $2 \times 2$ matrices.
Since the focus of this paper is on spin-1 systems, 
we need a basis of $3 \times 3$ matrices. 
There are several popular bases of $3 \times 3$ matrices, such as the Gell-Mann matrices~\cite{Gellmann1962} and the generalized Pauli matrices~\cite{Sylvester1882, Weyl1927}.
Here, however, we introduce an alternative basis $\{I,E_{0},E_{\pm 1},F_{0},F_{\pm 1},F_{\pm 2}\}$ as
\eq{
\label{eq:basis}
E_0 &= \begin{pmatrix}1&0&0\\0&0&0\\0&0&-1\end{pmatrix} \quad \left(=S^z\right)~,\nx
E_{+1} &= \begin{pmatrix}0&1&0\\0&0&1\\0&0&0\end{pmatrix} \quad \left(=\frac{S^x+\im S^y}{\sqrt{2}}\right)~,\nx
F_0 &= \begin{pmatrix}1&0&0\\0&-2&0\\0&0&1\end{pmatrix} \quad \left(=-(S^x)^2-(S^y)^2+2(S^z)^2\right)~,\nx
F_{+1} &= \begin{pmatrix}0&1&0\\0&0&-1\\0&0&0\end{pmatrix} \quad \left(=\left\{ S^z, \frac{S^x+\im S^y}{\sqrt{2}} \right\}\right)~,\nx
F_{+2} &= \begin{pmatrix}0&0&1\\0&0&0\\0&0&0\end{pmatrix} \quad \left(=\left(\frac{S^x+\im S^y}{\sqrt{2}} \right)^2\right)~,\nx
E_{-1} &= (E_{+1})\dg ~,\quad F_{-1} = (F_{+1})\dg ~,\quad F_{-2} = (F_{+2})\dg ~,
}
which are of the same form as the noncommutative version of spherical harmonic functions~\cite{Madore1992} up to constant factors.
It can be verified that the above matrices form a complete set, i.e.,
any $3 \times 3$ matrix can be expressed as a complex linear combination of these matrices.

\begin{table*}[t!]
\centering
 \caption{\label{tbl:CommutatorTable} The commutators $[a, b]$ where $a$ and $b$ are elements of our operator basis~\eqref{eq:basis}. 
  Each of these commutators is proportional to a single element of this operator basis 
  and satisfies Eq.~\eqref{eq:commutator_property}. 
 }
\def\arraystretch{1.5}
\begin{tabular}{|c||c|c|c|c|c|c|c|c|} \hline
   \diagbox[dir=NW]{a}{b} & $\ep$ & $\ez$ & $\em$ & $\fpp$ & $\fp$ & $\fz$ & $\fm$ & $\fmm$ \\ \hline\hline
    $\ep$&$0$&$-\ep$&$+\ez$&$0$&$-2\fpp$&$-3\fp$&$+\fz$&$+\fm$\\ \hline
    $\ez$&$+\ep$&$0$&$-\em$&$+2\fpp$&$+\fp$&$0$&$-\fm$&$-2\fmm$\\ \hline
    $\em$&$-\ez$&$+\em$&$0$&$-\fp$&$-\fz$&$+3\fm$&$+2\fmm$&$0$\\ \hline
    $\fpp$&$0$&$-2\fpp$&$+\fp$&$0$&$0$&$0$&$-\ep$&$+\ez$\\\hline
    $\fp$&$+2\fpp$&$-\fp$&$+\fz$&$0$&$0$&$-3\ep$&$+\ez$&$-\em$\\\hline
    $\fz$&$+3\fp$&$0$&$-3\fm$&$0$&$+3\ep$&$0$&$-3\em$&$0$\\\hline
    $\fm$&$-\fz$&$+\fm$&$-2\fmm$&$+\ep$&$-\ez$&$+3\em$&$0$&$0$\\\hline
    $\fmm$&$-\fm$&$+2\fmm$&$0$&$-\ez$&$+\em$&$0$&$0$&$0$\\\hline
  \end{tabular}
\end{table*}

This basis is more suitable than other bases for proving the nonintegrability of the BLBQ model 
because it satisfies the following properties:
(i) The commutator of any two base elements is proportional to some base element, unless they are commutative.
(ii) The BLBQ Hamiltonian is described in a relatively simple form.
\change{Note} that the Gell-Mann matrices do not satisfy the former and the generalized Pauli matrices do not satisfy the latter.
To briefly explain the former property, 
the commutation relations in the basis~\eqref{eq:basis} are described by
\eq{
[E_{m_1},E_{m_2}] &\propto E_{m_1+m_2} \nonumber \\
[E_{m_1},F_{m_2}] &\propto F_{m_1+m_2} \nonumber \\
[F_{m_1},F_{m_2}] &\propto E_{m_1+m_2}~,
\label{eq:commutator_property}}
where the right-hand side is zero if the symbol on it is not included in $\{I,E_{0},E_{\pm 1},F_{0},F_{\pm 1},F_{\pm 2}\}$.
Detailed commutation relations with constants of proportionality are displayed in Table~\ref{tbl:CommutatorTable}.

Using this basis, we rewrite the Hamiltonian~\eqref{eq:H_BLBQ} as
\eq{\label{eq:H_proof}
H&=\sumi \left(\sum_{m\in\{0,\pm 1\}}e_{m}E_{m,i}E_{-m,i+1} \right.
\nx
&\left. \qquad\qquad +\sum_{m\in\{0,\pm 1,\pm 2\}}f_{m}F_{m,i}F_{-m,i+1}+hF_{0,i} \right) ~,
}
where the coupling constants and field are
\eq{
e_{0}&=e_{\pm 1}=J_{1}-\frac{J_{2}}{2}~, \nx
    f_{0}&=\frac{J_{2}}{6}~,\quad f_{\pm 1}=\frac{J_{2}}{2}~,\quad f_{\pm 2}=J_{2}~,\nx
    h &= \frac{D}{3},
    \label{eq:couple_const_BLBQ}
}
with constant energy shift.
After proving the nonintegrability with this parameter set in Sec.~\ref{sec:proof}, we will treat in Sec.~\ref{sec:necessary} an extended model where coupling constants $\{e_m\}$ and $\{f_m\}$ take general values.

\subsection{Proof outline}
The idea of the proof of nonintegrability, introduced by Shiraishi's paper~\cite{Shiraishi2019}, is simple.
The nonintegrability 
(i.e.,  the absence of local conserved quantity) is proved by the following procedure:
First, we expand general local quantity $Q$ in a suitable basis and identify $Q$ with its expansion coefficient $\{q_\bullet\}$.
Then, we solve equations of $\{q_\bullet\}$ corresponding to conservation  condition $[Q,H]=0$ and prove the absence of (nontrivial) solutions.

As a basis of local quantities on spin-$1$ chain,
we take a subset of the basis of all quantities on spin-$1$ chain \eq{\bigotimes_{i=1}^N \{I_i, E_{0,i}, E_{\pm 1,i}, F_{0,i}, F_{\pm 1,i}, F_{\pm 2,i}\}\label{eq:chain_basis}}
with finitely restricted length of nontrivial consecutive support.
In particular, the basis of $k$-local quantity is $l$-local operators with $l\leq k$ in the basis~\eqref{eq:chain_basis}.
Any $k$-local quantity can be described by
\eq{
Q = \sum_{l=0}^k \sum_{\boldA^l_i} q_{\boldA^l_i} \boldA^l_i
\label{eq:Q}~,
}
where $\boldA^l_i$ represents an operator whose nontrivial consecutive support is $l$ and leftmost site is $i$.
Examples of $\boldA^3_i$ are $(\ez \ez \fz)_i := E_{0,i} E_{0,i+1} F_{0,i+2}$ and $(\fpp I \em)_i$.

Similarly, we can also expand $[Q,H]$ as
\eq{
[Q,H] = \sum_{l=0}^{k+1} \sum_{\boldA^l_i} r_{\boldB^l_i} \boldB^l_i~.
}
Here, we use the fact that the commutator of $k$-local quantity $Q$ and 2-local quantity $H$ is an at-most-$(k+1)$-local quantity.
Notice that each $r_{\boldB^l_i}$ is a linear combination of $\{q_{\boldA^l_i}\}$, and that the conservation condition $[Q,H]=0$ means $r_{\boldB^l_i}=0$ for all $\boldB^l_i$.
Therefore, we obtain a system of linear equations of $\{q_{\boldA^l_i}\}$.
The goal of nonintegrability proof is to prove the absence of nontrivial solutions to this system of equations.

We prove the absence of $k$-local conserved quantities with general $k \in [3,N/2]$ by the following two steps: 
In Step 1, we focus on the conditions$\{r_{\bkpi} = 0\}$ and derive that most of $\boldA^k_i$ have zero coefficients and that the remaining coefficients are proportional to each other.
In Step 2, we turn to the conditions $\{ q_{\boldB_j^k} = 0\}$ and derive that the remaining $\boldA^k_i$'s also have zero coefficients.
\change{This contradicts the assumption that $Q$ is a $k$-local quantity and proves the absence of $k$-local conserved quantities.}

\section{Proof}\label{sec:proof}
\begin{table*}[t!]
\centering
\def\arraystretch{1.5}
    \centering
    \caption{
    Table indexing the integrability/nonintegrability proofs of the BLBQ model.}
    \begin{tabular}{wc{5em}wc{5em}wc{5em}wc{10em}wc{10em}}\hline
        \multicolumn{3}{c}{Conditions} & Integrability & Proof \\ \hline
        \multirow{7}{*}{$J_2 \neq 0$} &\multicolumn{2}{c}{$J_1/J_2 \notin \{0,\pm 1, 1/2\}$} & Nonintegrable & Sec.~\ref{sec:proof} (here) \\ \cline{2-5}
        & \multicolumn{2}{c}{$J_1/J_2=1/2$} & Nonintegrable & Appendix~\ref{sec:singular_case}\\ \cline{2-5}
        & \multirow{2}{*}{$J_1=0$} &$D=0$ & Integrable & Refs.~\cite{Parkinson1987,Klumper1989,Barber1989}\\ \cline{3-5}
        & & $D\neq 0$ & Nonintegrable & Sec.~\ref{sec:necessary}\\ \cline{2-5}
        & \multicolumn{2}{c}{$J_1=J_2$} & Integrable & Refs.~\cite{Uimin1970,Lai1974,Sutherland1975,Schmitt1996}\\ \cline{2-5}
        & \multirow{2}{*}{$J_1=-J_2$} &$D=0$ & Integrable & Refs.~\cite{Takhatajan1982,Babujian1982} \\ \cline{3-5}
        & & $D\neq 0$ & Nonintegrable & Sec.~\ref{sec:necessary}\\ \hline
        \multirow{2}{*}{$J_2 = 0$} & \multicolumn{2}{c}{$J_1 = 0$} & Integrable & (Trivial) \\ \cline{2-5}
         & \multicolumn{2}{c}{$J_1 \neq 0$} & Nonintegrable & Appendix~\ref{sec:singular_case} \\ \hline
    \end{tabular}
    \label{tbl:where_proof}
\end{table*}

We prove the nonintegrability of the BLBQ  
[\eqref{eq:H_proof} and~\eqref{eq:couple_const_BLBQ}], 
except for known integrable systems, i.e., ones that satisfies Eq.~\eqref{eq:except}. The guide to the proof is given in Table~\ref{tbl:where_proof}. In this section, we deal with a generic case where $J_2 \neq 0$ and $J_1/J_2 \notin \{0,\pm1,1/2\}$ hold. 
In this case, all the coupling constants given in Eq.~\eqref{eq:couple_const_BLBQ} have nonzero values, as well as the nonintegrability holds regardless of the value of the magnetic field $D$.
The proof of the case where some coupling constants are zero is given in Appendix~\ref{sec:singular_case};
The proof of the case where integrability and nonintegrability switch with the value of $D$ is given in Sec.~\ref{sec:necessary}.

The structure of this section is as follows: In Secs.~\ref{subsec41} and \ref{subsec42}, we give a proof of the absence of 3-local conserved quantities, in order to grasp a picture of our proof.
In Secs.~\ref{subsec43} and \ref{subsec44}, we prove the absence of $k$-local conserved quantities for general $k \in [3,N/2]$, which completes the proof of Theorem~\ref{thm:main} for the generic case. 
In Secs.~\ref{subsec45} and \ref{subsec46}, we conduct an analysis of 2-local and 1-local conserved quantities, respectively, which corresponds to Theorem~\ref{thm:1and2}.

\subsection{Step 1 (analysis of $r_{{\bf B}^{k+1}_i} = 0$) for $k=3$ case}\label{subsec41}
Using the fact that coefficients of the 4-local operators appearing in the commutator are zero: $\{ r_{\Bkpi} = 0 \}$, we show that most of the coefficients $ q_{\Aki} $, 
which correspond to the operators with the largest support, 
are zero.
We also prove that the remaining coefficients are all proportional 
and therefore there is only 1 degree of freedom.

For example, consider the commutator between 
$\Aki = E_{0,i} E_{0,i+1} E_{+1,i+2}$ and $E_{-1,i+2} E_{+1,i+3}$, 
which are the terms in the conserved quantity $Q$ and the Hamiltonian $H$, respectively. 
Since these two operators share only the ($i+2$)th site, 
this commutator can be calculated as
\eq{
[E_{0,i} E_{0,i+1} E_{+1,i+2}, E_{-1,i+2} E_{+1,i+3}] \nx= + E_{0,i} E_{0,i+1} E_{0,i+2} E_{+1,i+3}~,\label{eq:example_3local}
}
where we used $[\ep,\em]=\ez$ displayed in Table~\ref{tbl:CommutatorTable}. 
We represent this relation as
\eq{
\bmat{
 & \ez & \ez & \ep \\
 &     &     & \em & \ep \\ \hline
+& \ez & \ez & \ez & \ep
}~,\label{eq:column_expression}
}
which we refer to as {\itshape a column expression of a commutator}. 
In a column expression, the first row represents a term of $Q$, 
the second row represents a term of the Hamiltonian $H$,
and the last row represents the commutator of these two terms.
The horizontal line depicts commutation operation. 
These three operators are placed to reflect the position of the sites 
on which they act. 
The column expression~\eqref{eq:column_expression} represents the same information 
as Eq.~\eqref{eq:example_3local}, 
except that the site number is omitted in Eq.~\eqref{eq:column_expression}.
Our aim is to obtain an equation of $\{q_\Aki\}$ that corresponds to $r_\Bkpi=0$.
Therefore, we need to search for other terms in $Q$ and $H$
whose commutators yield the output $\Bkpi = (\ez \ez \ez \ep)_i := E_{0,i} E_{0,i+1} E_{0,i+2} E_{+1,i+3}$. 
Fortunately, we find that no such terms exist except in Eq.~\eqref{eq:column_expression}.
This directly means that
\eq{
\cem q_{(\ez \ez \ep)_i} = + r_{(\ez \ez \ez \ep)_i} = 0\label{eq:coefficient_ZZE}
}
holds for all $i$.
Since we assume $\cem\neq0$, 
we know from Eq.~\eqref{eq:coefficient_ZZE} that the coefficient of $\Aki = (\ez \ez \ep)_i$ in $Q$ is zero for all $i$. 

In general, our task is finding out all the commutators that produce some fixed $\Bkpi$ 
and obtaining an equation of coefficients $\{q_\Aki\}$. 
The key observation is that for any $\Bkpi$, 
there are at most two commutators that generate it, 
consisting of a $3$-local term in $Q$ and a $2$-local term in $H$. 
The following two are only the commutators that may produce a 4-local operator $\Bkpi = (\Bone \Btwo \Bthree \Bfour)_i$:
\eq{
\bmat{
\Bone & \Btwo & ? \\
  &     & \Bfour^\dagger & \Bfour \\ \hline
\Bone & \Btwo & \Bthree & \Bfour
}\qquad
\bmat{
& ? & \Bthree & \Bfour\\
  \Bone & \Bone^\dagger     \\ \hline
\Bone & \Btwo & \Bthree & \Bfour
}~,\label{eq:two_column}
}
where we use the fact that 
all interaction terms in our Hamiltonian~\eqref{eq:H_proof}
are written as $C_i C^\dagger_{i+1}$ for some operator basis element $C$. 
Thus, 
at most one commutator corresponds to each of the left and right forms in Eq.~\eqref{eq:two_column}.

If only one commutator generates $\Bkpi$, 
we can immediately conclude that $\Aki$ has zero coefficients, as in Eq.~\eqref{eq:coefficient_ZZE}. 
The following proposition provides the sufficient conditions for such a case.
\begin{prop}
\label{prop:almost_dp_3}
Assume $Q$ to be a 3-local conserved quantity and include $\Aki = (\Aone \Atwo \Athree)_i$ with coefficient $q_\Aki$.
If there is no operator basis element $C$ satisfying $\Atwo \propto [\Aone\dg, C]$, 
then $q_\Aki=0$ holds for any $i$.
\end{prop}
This proposition means that 
a candidate $3$-local conserved quantity consists only of terms 
with the form of $(\Aone [\Aone\dg, C] \Athree)_i$ 
and $k(\leq2)$-local terms. 
\begin{proof}
We fix the site $i$ and omit writing it in the following.
According to Table~\ref{tbl:CommutatorTable}, 
for any element $\Athree$ of the operator basis, 
there is at least one element of an operator basis $D$ that does not commute with $\Athree$. 
As discussed above, 
there are at most two commutators that provide $\Bkp \propto \Aone \Atwo [\Athree,\tilde{C}] {\tilde{C}}^\dagger$. 
These commutators can be written in the following form:
\eq{
\def\arraystretch{1.3}
\bmat{
 & \Aone & \Atwo & \Athree \\
 &     &     & \tilde{C} & {\tilde{C}}^\dagger \\ \hline
 & \Aone & \Atwo & [\Athree,{\tilde{C}}] & {\tilde{C}}^\dagger
}\qquad
\bmat{
 &         & ?       & \Bthree & {\tilde{C}}^\dagger \\
 & \Aone & \Aone^\dagger \\ \hline
 & \Aone & \Atwo & \Bthree & {\tilde{C}}^\dagger
}~.
\def\arraystretch{1}
}
By the assumption of the proposition, however, there is no operator corresponding to the right form. 
Therefore, we have the conclusion: $q_{(\Aone \Atwo \Athree)_i} = 0$.
\end{proof}

We introduce {\itshape a column expression of an operator} as 
$A_i([B , C])_{i+1}D_{i+2} = \bmatc{\hline A & B &\\& C&D \\\hline}{|ccc|}_{~i}$. 
For example, we have $\bmatc{\hline \fz \\ \em \\\hline}{|c|} = -3\fm$
and $\bmatc{\hline \fz & \fz &\\ & \em & \ep \\\hline}{|ccc|} = -3\fz \fm \ep$. 
We also use the same symbol for the $k\geq4$-local operators. 
Operators of this form coincide with the elements of the local quantity basis, up to constant factors.
For notational simplicity, we will also treat these operators as elements of the local quantity basis. 
Here, the expansion coefficients are scaled so that the following relation holds:
\eq{
q_{\tilde{\boldA}^k_i} \tilde{\boldA}^k_i = q_\aki \aki~. 
}
For example, we have $q_{(-3 \fz \fm \ep)_i} := -\frac{1}{3} q_{(\fz \fm \ep)_i}$. 
By using this expression, 
the conclusion of Proposition~\ref{prop:almost_dp_3} can be understood 
that the candidates for $\Ak$ with nonzero coefficients are reduced to the following form:
\eq{
\Ak =
\bmatc{
\hline
\Aone & \Aone \dg& \\
& C & \Athree\dg \\ \hline
}{|ccc|}~.\label{eq:almost doubling_3}
}
The following proposition restricts the candidates for $\Ak$ with nonzero coefficients to the case $C = \Athree$.
\begin{prop}
Assume $Q$ to be a $3$-local conserved quantity and include $\Aki$ described by Eq.~\eqref{eq:almost doubling_3}.
Unless 
$C=\Athree$, then $q_\Aki=0$ holds. 
\label{prop:dp_3}
\end{prop}
\begin{proof}
First, we prove the proposition for the case of $\Athree = E_0$. 
That is, we show that a coefficient of $\Ak = \bmatc{
\hline
\Aone & \Aone \dg& \\
& C & \ez \\ \hline
}{|ccc|}$ is zero if $C \neq E_0$. 
For the case of $C = F_0$, 
we consider an output $\Bkp = \bmatc{\hline \Aone & \Aone\dg & & \\ & \fz & \fpp & \fmm \\\hline}{|cccc|}_{~i}$ 
and obtain the following two commutators:
\eq{
\bmatc{
\cline{2-4}
 & \Ic{\Aone} & \Aone\dg & \cI{} \\
 & \Ic{     } & \fz & \cI{\ez} \\\cline{2-4}
 &     &     & \fpp & \fmm \\ \hline \\[-0.8em]
 \cline{2-5}
+2& \Ic{\Aone} & \Aone\dg &          & \cI{} \\
 & \Ic{}      & \fz    & \fpp & \cI{\fmm} \\
\cline{2-5}
}{ccccc} \quad
\bmatc{
\cline{3-5}
 &         & \Ic{\fz}   & \ez & \cI{} \\
 &         & \Ic{}   & \fpp & \cI{\fmm} \\ \cline{3-5}
 & \Aone & \Aone^\dagger \\ \hline \\[-0.8em]
\cline{2-5}
-& \Ic{\Aone} & \Aone\dg &          & \cI{} \\
 & \Ic{}      & \fz    & \fpp & \cI{\fmm} \\
\cline{2-5}
}{ccccc}.
\label{eq:two_comm}
}
Thus, by the condition $r_\Bkpi=0$, 
we have 
\eq{
&+2 q~\bmatc{
\hline
\Aone & \Aone \dg& \\
& \fz & \ez \\ \hline
}{|ccc|}_{~i}
-
q~\bmatc{
\hline
\fz & \ez & \\
& \fpp & \fmm \\ \hline
}{|ccc|}_{~i+1} \nx
&= r~\bmatc{\hline \Aone & \Aone\dg & & \\ & \fz & \fpp & \fmm \\\hline}{|cccc|}_{~i} = 0~.
}
Therefore, we find that 
the coefficient of 
$\Ak = \bmatc{
\hline
\Aone & \Aone \dg& \\
& \fz & \ez \\ \hline
}{|ccc|}$
is proportional to the coefficient of 
$\bmatc{
\hline
\fz & \ez & \\
& \fpp & \fmm \\ \hline
}{|ccc|}$:
\eq{
q~\bmatc{
\hline
\Aone & \Aone \dg& \\
& \fz & \ez \\ \hline
}{|ccc|}_{~i}
\propto
q~\bmatc{
\hline
\fz & \ez & \\
& \fpp & \fmm \\ \hline
}{|ccc|}_{~i+1}~,
\label{eq:not_doubling_3}
}
where the proportionality constant is nonzero,  
which is determined by the coupling constants of the Hamiltonian and the constant factor in the commutation relation 
(see Table~\ref{tbl:CommutatorTable}). 
On the other hand, 
the coefficient of 
$\bmatc{
\hline
\fz & \ez & \\
& \fpp & \fmm \\ \hline
}{|ccc|} \propto \fz  \fpp  \fmm$
is zero by Proposition~\ref{prop:almost_dp_3}. 
Therefore, we have shown that 
the coefficient $q_\Aki$ of 
$\Ak = \bmatc{
\hline
\Aone & \Aone \dg& \\
& \fz & \ez \\ \hline
}{|ccc|}$
is zero. 

Similarly, 
for the case of $C \neq \ez,\fz$, 
we have 
\eq{
q~\bmatc{
\hline
\Aone & \Aone \dg& \\
& C & \ez \\ \hline
}{|ccc|}_{~i}
\propto
q~\bmatc{
\hline
C & \ez & \\
& C & C\dg \\ \hline
}{|ccc|}_{~i+1}
= 0~.
}
Here, the equality follows from Proposition~\ref{prop:almost_dp_3} 
because $\Ak = \bmatc{
\hline
C & \ez & \\
& C & C\dg \\ \hline
}{|ccc|} \propto CCC\dg$ cannot be written in the form of Eq.~\eqref{eq:almost doubling_3}, i.e., 
$\Ak\not\propto \bmatc{
\hline
C& C\dg & \\ 
 & \tilde{C} & C\dg\\ \hline
}{|ccc|}$ for any $\tilde{C}$. Absence of such a $\tilde{C}$ can be confirmed in Table~\ref{tbl:CommutatorTable}. 
This completes the proof for the case of $\Athree = E_0$. 
In other words, we obtain
\eq{ q_{(\Aone \Atwo \ez)_i} = 0 \quad \text{except for }\Atwo \propto [\Aone\dg\,\ez]~.
 \label{eq:almost_dp_z_3}
 }

For general $\Athree$, 
we reduce the problem to the case of $\Athree = E_0$. 
If $\Athree \neq \ez, \fz$, 
by considering column relations similar to Eq.~\eqref{eq:two_comm},
we have
\eq{
q~\bmatc{
\hline
\Aone & \Aone \dg& \\
& C & \Athree\dg \\ \hline
}{|ccc|}_{~i}
\propto
q~\bmatc{
\hline
C & \Athree\dg & \\
& \ez & \ez \\ \hline
}{|ccc|}_{~i+1}~.
\label{eq:prf_of_doubling1}}
Because of $\bmatc{
\hline
C & \Athree\dg & \\
& \ez & \ez \\ \hline
}{|ccc|}
\propto C\Athree\dg\ez$ and Eq.~\eqref{eq:almost_dp_z_3}, 
the right hand side of Eq.~\eqref{eq:prf_of_doubling1} is zero unless $C=\Athree$. 

For the case of $\Athree = \fz$, we have
\eq{
q~\bmatc{
\hline
\Aone & \Aone \dg& \\
& C & \fz \\ \hline
}{|ccc|}_{~i}
\propto
q~\bmatc{
\hline
C & \fz & \\
& \ep & \em \\ \hline
}{|ccc|}_{~i+1}~.
}
From Proposition~\ref{prop:almost_dp_3}, for $q~\bmatc{
\hline
C & \fz & \\
& \ep & \em \\ \hline
}{|ccc|}_{~i+1}
$ to be nonzero, 
it is necessary that 
$
\bmatc{
\hline
C & \fz & \\
& \ep & \em \\ \hline
}{|ccc|}
=
\bmatc{
\hline
C & C\dg & \\
& \tilde{C} & \em \\ \hline
}{|ccc|}
$ holds for some $\tilde{C}$. 
By the proof of this proposition for the case of $\Athree=\em$, 
$q~\bmatc{
\hline
C & \fz & \\
& \ep & \em \\ \hline
}{|ccc|}_{~i+1}
$ is zero unless $\tilde{C}=\ep$. 
This condition $\tilde{C}=\ep$ is equivalent to $C (=C\dg) = \fz=\Athree$. 
This concludes the proof of Proposition~\ref{prop:dp_3}.  
\end{proof}
We have seen that 
the $3$-local operators that may have nonzero coefficients are restricted to the following form:
$
\Ak = \bmatc{
\hline
\Cone & \Cone\dg & \\
& \Ctwo & \Ctwo\dg \\ \hline
}{|ccc|}
$, 
which we call a {\it doubling operator}. 
Here $\Cone$ and $\Ctwo$ are arbitrary elements of the operator basis that are noncommutative.

As the last proposition of Step 1, 
we state that the coefficients of the three-support operators $\Aki$ have only 1 degree of freedom. 
\begin{prop}
Asuume $Q$ to be a $3$-local conserved quantity. 
The coefficients of all doubling operators are proportional to each other. 
Precisely, their coefficients can be written as
\eq{
q~\bmatc{
\hline
\Cone & \Cone\dg & \\
& \Ctwo & \Ctwo\dg \\ \hline
}{|ccc|}_{~i} = q_{k=3} c_{(1)} c_{(2)}~,
\label{eq:dp_factor_3}
}
by using some constant $q_{k=3}$ independent of $i$. 
Here, we denote the coupling constant of the interaction term $C_{(n),j} C_{(n),j+1}\dg$ as $c_{(n)}$.
\label{prop:dp_factor_3}
\end{prop}
\begin{proof}
By similar arguments as before, 
the coefficients of doubling operators are proportional to each other as
\eq{
 q~\bmatc{
\hline
\Cone & \Cone\dg & \\
& \Ctwo & \Ctwo\dg \\ \hline
}{|ccc|}_{~i} 
&\propto q~\bmatc{
\hline
\Ctwo & \Ctwo\dg & \\
& \Cthree & \Cthree\dg \\ \hline
}{|ccc|}_{~i+1} \nx
&\propto q~\bmatc{
\hline
\Cthree & \Cthree\dg & \\
& \Cfour & \Cfour\dg \\ \hline
}{|ccc|}_{~i+2} \nx
&\propto \cdots~,
}
where the sequence $\{\Cthree, \Cfour, \dots\}$ is arbitrary 
as long as $C_{(n-1)}\dg$ and $C_{(n)}$ are noncommutative.
By taking $\{C_{(n)}\}$ appropriately,
we can show that the coefficients of any two doubling operators are proportional.

Next, we determine the proportionality constants of the coefficients.  
We consider two commutators with the following general forms:
\eq{
\bmatc{
\cline{2-4}
 & \Ic{\Cone} & \Cone\dg & \cI{} \\
 & \Ic{     } & \Ctwo & \cI{\Ctwo\dg} \\\cline{2-4}
 &     &     & \Cthree & \Cthree\dg \\ \hline \\[-0.8em] \cline{2-5}
+& \Ic{\Cone} & \Cone\dg &          & \cI{} \\
 & \Ic{}      & \Ctwo    & \Ctwo\dg & \cI{} \\
 & \Ic{}      &          & \Cthree  & \cI{\Cthree\dg}\\ \cline{2-5}
}{ccccc} \qquad
\bmatc{
\cline{3-5}
 &       & \Ic{\Ctwo} & \Ctwo\dg & \cI{} \\
 &       & \Ic{}      & \Cthree  & \cI{\Cthree\dg} \\\cline{3-5}
 & \Cone & \Cone\dg    &  &  \\ \hline \\[-0.8em] \cline{2-5}
-& \Ic{\Cone} & \Cone\dg &          & \cI{} \\
 & \Ic{}      & \Ctwo    & \Ctwo\dg & \cI{} \\
 & \Ic{}      &          & \Cthree  & \cI{\Cthree\dg}\\ \cline{2-5}
}{ccccc}~,
}
where we define 
\eq{
\bmatc{\hline \Cone & \Cone\dg & & \\ & \Ctwo & \Ctwo\dg & \\ & & \Cthree &\Cthree\dg \\ \hline}{|cccc|} 
:= \Cone \bmatc{\hline \Cone\dg \\ \Ctwo \\\hline}{|c|}~\bmatc{\hline \Ctwo\dg \\ \Cthree \\\hline}{|c|} \Cthree\dg~,\label{eq:general_column_expression}
} 
as a generalization of the column expression of operators. 
The above two forms yield
\eq{
c_{(3)} q~\bmatc{
\hline
\Cone & \Cone\dg & \\
& \Ctwo & \Ctwo\dg \\ \hline
}{|ccc|}_{~i}
-c_{(1)} q~\bmatc{
\hline
\Ctwo & \Ctwo\dg & \\
& \Cthree & \Cthree\dg \\ \hline
}{|ccc|}_{~i+1} = 0 ~, 
}
which is equivalent to Eq.~\eqref{eq:dp_factor_3}. 
\end{proof}

\subsection{Step 2 (analysis of $r_{{\bf B}^{k}_i} = 0$) for $k=3$ case}\label{subsec42}
So far, we have only focused on $4$-local terms of $[Q,H]$.
In this step, by focusing on $3$-local outputs, 
we eliminate the last 1 degree of freedom $q_{k=3}$ and show that all $q_{\Aki} $ are zero. 
In fact, in the case $k=3$, 
this process is completed by focusing only on $r_{(\ez \ez \fz)_i} = 0$. 
The set of commutators that can yield $\Bk = \ez \ez \fz$ are listed as follows:
\begin{widetext}
\eq{
\bmat{
 & ? & \fz \\
\ez   & \ez  \\ \hline
\ez & \ez & \fz
}\qquad
\bmat{
\ez & ? & \\
    & \fz & \fz \\ \hline
\ez & \ez & \fz
}\quad \left| \quad
\bmat{
? & \ez & \fz \\
\fz \\ \hline
\ez & \ez & \fz
}\qquad
\bmat{
\ez & ? & \fz \\
&\fz \\ \hline
\ez & \ez & \fz
}\qquad
\bmat{
\ez & \ez & ? \\
&&\fz \\ \hline
\ez & \ez & \fz
}\quad \right| \quad
\bmat{
? & ? & \fz \\
     ? & ? \\ \hline
\ez & \ez & \fz
}\qquad
\bmat{
\ez & ? & ? \\
    & ? & ? \\ \hline
\ez & \ez & \fz
}~.
\label{eq:zzw_q_3}
}
\end{widetext}
The left two in Eq.~\eqref{eq:zzw_q_3} are commutators of $2$-local terms in $Q$ and the interaction terms in $H$, generating $\ez \ez \fz$.
Such a form does not exist in fact, 
because $[C,\ez]$ and $[C,\fz]$ are not proportional to $\ez$ for any element $C$ in the operator basis. 
Next, the middle three forms in Eq.~\eqref{eq:zzw_q_3} represent 
commutators of $3$-local terms in $Q$ and on-site term $\fz$ in $H$. 
No such a form exists either, 
because $[C,\fz]$ are not proportional to $\ez$ or $\fz$ for any element $C$ in the operator basis. 
Therefore, all that yields nontrivial terms are 
the remaining right two forms in Eq.~\eqref{eq:zzw_q_3}, 
that is, commutators
of $3$-local terms in $Q$ and the interaction terms in $H$. 
These operators share two sites and 
therefore have a more complex commutator than before.
For example, we consider the commutator between 
$E_{+1,i} F_{-1,i+1} F_{-1,i+2}$ and $E_{-1,i} E_{+1,i+1}$. 
By using 
the anticommutator Table~\ref{tbl:AnticommutatorTable}
in Appendix~\ref{sec:multiplication_table}, 
we have
\eq{
&[E_{+1,i} F_{-1,i+1}, E_{-1,i} E_{+1,i+1}] \nx
&= \frac{1}{2}E_{0,i}\{F_{-1,i+1},E_{+1,i+1}\} 
- \frac{1}{2}\{E_{-1,i},E_{+1,i}\}F_{0,i+1} \nx
&=\frac{1}{2}E_{0,i}E_{0,i+1}
-\frac{1}{2}
\left(\frac{4}{3} I_i - \frac{1}{3}F_{0,i} \right)F_{0,i+1}\nx
&= \frac{1}{2} E_{0,i} E_{0,i+1} + \frac{1}{6} F_{0,i} F_{0,i+1} - \frac{2}{3} F_{0,i+1}~,
\label{eq:efee_s_3}
}
\change{where $\{\cdot,\cdot\}$ is the anticommutator. 
Here, we used the following equality:
\eq{
&[A\otimes B,C\otimes D]\nx
&\quad=\frac{1}{2}[A,C]\otimes\{B,D\}+\frac{1}{2}\{A,C\}\otimes [B,D].
}}
As can be seen from these equations, a commutator of operators that share more than one site may not necessarily be written with a single operator in our basis.
We use a double line in the column expression of the commutator, 
meaning that we only focus on one of the output operators (here, $\Bki = E_{0,i} E_{0,i+1} F_{0,i+2}$) as
\eq{
\bmatc{
&\ep & \fm & \fz \\
&\em & \ep & \\ \hline \hline
+1/2&\ez & \ez & \fz
}{cccc}~.}
We also combine this symbol with the column expression of an operator as
\eq{
\bmatc{
\cline{2-4}
&\Ic{\ep} & \em & \cI{} \\
&\Ic{}    & \fz & \cI{\fz} \\ \cline{2-4}
&\em & \ep & \\ \hline \hline
+3/2&\ez & \ez & \fz
}{cccc}~.
\label{eq:efee_g_3}
}

In the following, 
we list all contributions to $\ez \ez \fz$ that can be written in the two forms on the right of Eq.~\eqref{eq:zzw_q_3}. 
Noting that $3$-local terms $\Ak$ in $Q$ can be written in the form of doubling operator, 
these forms are either of the following:
\eq{
\bmat{
\hline
\Ic{C} & C\dg & \cI{} \\
\Ic{} & \fz & \cI{\fz} \\\hline
     D & D\dg \\ \hline\hline
\ez & \ez & \fz
}\qquad
\bmat{
\hline
\Ic{\ez} & \ez & \cI{} \\
\Ic{} & C & \cI{C\dg} \\\hline
    & D & D\dg \\ \hline\hline
\ez & \ez & \fz
}~,
\label{eq:cd_3}
}
where $C$ and $D$ are some elements of the operator basis. 
Analogous to Eq.~\eqref{eq:commutator_property}, 
the products of two operators 
$E_{m_1}E_{m_2}, E_{m_1}F_{m_2}, F_{m_1}E_{m_2}$, and $F_{m_1}F_{m_2}$ can be written by linear combinations of $E_{m_1+m_2}$ and $F_{m_1+m_2}$ 
where we consider $m_1,m_2\in\{0,\pm1,\pm2\}$~\footnote{
These relations can be seen from the fact that subscript $m$ represents the eigenvalue with respect to the adjoint action by $\ez$: $[\ez,\bullet]$.
}.
Thus, the candidate pairs of $C$ and $D$ that generate $\ez \ez \fz$ are 
$(C,D) = (E_{\pm 1} E_{\mp 1}), (E_{\pm 1} F_{\mp 1}), (F_{\pm 1} E_{\mp 1}), (F_{\pm 1} F_{\mp 1})$, and $ (F_{\pm 2} F_{\mp 2 })$. 
Following Eq.~\eqref{eq:efee_g_3}, 
we enumerate the contributions to $\ez \ez \fz$ of all these commutators as
\begin{widetext}
 \!\!\!\!\!\!\!\!
 \eq{
 \bmatc{
 \cline{2-4}
 &\Ic{\ep} & \em & \cI{} \\
 &\Ic{}    & \fz & \cI{\fz} \\ \cline{2-4}
 &\em & \ep & \\ \hline \hline
 +3/2&\ez & \ez & \fz
 }{cccc}
 \quad
 \bmatc{
 \cline{2-4}
 &\Ic{\ep} & \em & \cI{} \\
 &\Ic{}    & \fz & \cI{\fz} \\ \cline{2-4}
 &\fm & \fp & \\ \hline \hline
 -3/2&\ez & \ez & \fz
 }{cccc}
 \quad
 \bmatc{
 \cline{2-4}
 &\Ic{\fp} & \fm & \cI{} \\
 &\Ic{}    & \fz & \cI{\fz} \\ \cline{2-4}
 &\em & \ep & \\ \hline \hline
 -3/2&\ez & \ez & \fz
 }{cccc}
 \quad
 \bmatc{
 \cline{2-4}
 &\Ic{\fp} & \fm & \cI{} \\
 &\Ic{}    & \fz & \cI{\fz} \\ \cline{2-4}
 &\fm & \fp & \\ \hline \hline
 +3/2&\ez & \ez & \fz
 }{cccc}
 \quad
 \bmatc{
 \cline{2-4}
 &\Ic{\fpp} & \fmm & \cI{} \\
 &\Ic{}    & \fz & \cI{\fz} \\ \cline{2-4}
 &\fmm & \fpp & \\ \hline \hline
 0&\ez & \ez & \fz
 }{cccc} \nx \nx
 \!\!\!\!\!\!\!\!
 \bmatc{
 \cline{2-4}
 &\Ic{\em} & \ep & \cI{} \\
 &\Ic{}    & \fz & \cI{\fz} \\ \cline{2-4}
 &\ep & \em & \\ \hline \hline
 +3/2&\ez & \ez & \fz
 }{cccc}
 \quad
 \bmatc{
 \cline{2-4}
 &\Ic{\em} & \ep & \cI{} \\
 &\Ic{}    & \fz & \cI{\fz} \\ \cline{2-4}
 &\fp & \fm & \\ \hline \hline
 -3/2&\ez & \ez & \fz
 }{cccc}
 \quad
 \bmatc{
 \cline{2-4}
 &\Ic{\fm} & \fp & \cI{} \\
 &\Ic{}    & \fz & \cI{\fz} \\ \cline{2-4}
 &\ep & \em & \\ \hline \hline
 -3/2&\ez & \ez & \fz
 }{cccc}
 \quad
 \bmatc{
 \cline{2-4}
 &\Ic{\fm} & \fp & \cI{} \\
 &\Ic{}    & \fz & \cI{\fz} \\ \cline{2-4}
 &\fp & \fm & \\ \hline \hline
 +3/2&\ez & \ez & \fz
 }{cccc}
 \quad
 \bmatc{
 \cline{2-4}
 &\Ic{\fmm} & \fpp & \cI{} \\
 &\Ic{}    & \fz & \cI{\fz} \\ \cline{2-4}
 &\fpp & \fmm & \\ \hline \hline
 0&\ez & \ez & \fz
 }{cccc}~,\label{eq:zzw_3_from_left}
 }
 \eq{
 \!\!\!\!\!\!\!\!
 \bmatc{
 \cline{2-4}
 &\Ic{\ez} & \ez & \cI{} \\
 &\Ic{}    & \ep & \cI{\em} \\ \cline{2-4}
 &  & \em & \ep \\ \hline \hline
 -1/6&\ez & \ez & \fz
 }{cccc}
 \quad
 \bmatc{
 \cline{2-4}
 &\Ic{\ez} & \ez & \cI{} \\
 &\Ic{}    & \ep & \cI{\em} \\ \cline{2-4}
 &  & \fm & \fp \\ \hline \hline
 -1/2&\ez & \ez & \fz
 }{cccc}
 \quad
 \bmatc{
 \cline{2-4}
 &\Ic{\ez} & \ez & \cI{} \\
 &\Ic{}    & \fp & \cI{\fm} \\ \cline{2-4}
 &  & \em & \ep \\ \hline \hline
 -1/2&\ez & \ez & \fz
 }{cccc}
 \quad
 \bmatc{
 \cline{2-4}
 &\Ic{\ez} & \ez & \cI{} \\
 &\Ic{}    & \fp & \cI{\fm} \\ \cline{2-4}
 &  & \fm & \fp \\ \hline \hline
 -1/6&\ez & \ez & \fz
 }{cccc}
 \quad
 \bmatc{
 \cline{2-4}
 &\Ic{\ez} & \ez & \cI{} \\
 &\Ic{}    & \fpp & \cI{\fmm} \\ \cline{2-4}
 &  & \fmm & \fpp \\ \hline \hline
 +1/3&\ez & \ez & \fz
 }{cccc} \nx \nx
 \!\!\!\!\!\!\!\!
 \bmatc{
 \cline{2-4}
 &\Ic{\ez} & \ez & \cI{} \\
 &\Ic{}    & \em & \cI{\ep} \\ \cline{2-4}
 &  & \ep & \em \\ \hline \hline
 -1/6&\ez & \ez & \fz
 }{cccc}
 \quad
 \bmatc{
 \cline{2-4}
 &\Ic{\ez} & \ez & \cI{} \\
 &\Ic{}    & \em & \cI{\ep} \\ \cline{2-4}
 &  & \fp & \fm \\ \hline \hline
 -1/2&\ez & \ez & \fz
 }{cccc}
 \quad
 \bmatc{
 \cline{2-4}
 &\Ic{\ez} & \ez & \cI{} \\
 &\Ic{}    & \fm & \cI{\fp} \\ \cline{2-4}
 &  & \ep & \em \\ \hline \hline
 -1/2&\ez & \ez & \fz
 }{cccc}
 \quad
 \bmatc{
 \cline{2-4}
 &\Ic{\ez} & \ez & \cI{} \\
 &\Ic{}    & \fm & \cI{\fp} \\ \cline{2-4}
 &  & \fp & \fm \\ \hline \hline
 -1/6&\ez & \ez & \fz
 }{cccc}
 \quad
 \bmatc{
 \cline{2-4}
 &\Ic{\ez} & \ez & \cI{} \\
 &\Ic{}    & \fmm & \cI{\fpp} \\ \cline{2-4}
 &  & \fpp & \fmm \\ \hline \hline
+1/3&\ez & \ez & \fz
}{cccc}~.\label{eq:zzw_3_from_right} 
}
Therefore, applying the results of Proposition~\ref{prop:dp_factor_3} yields the following equation:
\eq{
0 = r_{(\ez \ez \fz)_i} = q_{k=3} \Bigl( \cfz (&+3/2 \cep \cem -3/2 \cep \cfm -3/2 \cfp \cem +3/2 \cfp \cfm +0 \cfpp \cfmm \nx
                                      &+3/2 \cem \cep -3/2 \cem \cfp -3/2 \cfm \cep +3/2 \cfm \cfp +0 \cfmm \cfpp) \nx
                                +\cez (&-1/6 \cep \cem -1/2 \cep \cfm -1/2 \cfp \cem -1/6 \cfp \cfm +1/3 \cfpp \cfmm \nx
                                      &-1/6 \cem \cep -1/2 \cem \cfp -1/2 \cfm \cep -1/6 \cfm \cfp +1/3 \cfmm \cfpp)\Bigr)~.
\label{eq:condition_3}}
By substituting Eq.~\eqref{eq:couple_const_BLBQ}, 
we obtain
\eq{
0 &= q_{k=3} ( -1/3 (J_1-J_2/2)^3 - 1/2 (J_1-J_2/2)^2 J_2 + 1/12 (J_1-J_2/2)J_2^2 + 1/8 J_2^3 ) \nx
&= -\frac{q_{k=3}}{3} J_1 (J_1 - J_2) (J_1 + J_2)~,
}
\end{widetext}
which means that $q_{k=3}=0$ holds unless $J_1 \neq 0, \pm J_2$. 
Therefore, 
together with exception handling in Appendix~\ref{sec:singular_case}  and Sec.~\ref{sec:necessary},
the BLBQ model~\eqref{eq:H_BLBQ} is shown to possess no $3$-local conserved quantity 
except for known integrable systems. 
This completes the proof of Theorem~\ref{thm:main} for $k=3$.

\subsection{Step 1 (analysis of $r_{{\bf B}^{k+1}_i} = 0$) for $3 \leq k \leq N/2$}\label{subsec43}
We give a proof of the absence of general $k$-local conserved quantities ($3 \leq k \leq N/2$). 
We note that the argument for $k=3$ is also essential for general $k$. 
In particular, 
Step~1 in the proof for general $k$ is parallel to that for $k=3$ (see Sec.~\ref{subsec41}). 
In the following, we use a generalization of the column expression of operators to multiple rows [e.g., for three rows, see Eq.~\eqref{eq:general_column_expression}].
\begin{prop}\label{prop:almost_dp_k}
Assume $Q$ to be a $k$-local conserved quantity and include $\Aki$ with coefficient $q_\Aki$.
If there is no operator basis elements $\{C_{(n)}\}$ satisfying
\eq{
\ak \propto 
\bmatc{
\hline
\Cone & \Cone\dg &          &        &       &          &\\
      & \Ctwo    & \Ctwo\dg &        &       &          &\\
      &          &          & \ddots &       &          &\\
      &          &          &        & \Ckmm & \Ckmm\dg &\\
      &          &          &        &       & \Ckm     & A_{(k)}\\\hline
}{|ccccccc|}~,
\label{eq:almost_dp_k}
}
then $q_\aki=0$ holds for any $i$. 
\end{prop}
\begin{proof}
The proof is similar to that of Proposition~\ref{prop:almost_dp_3}. 
First, we consider $\ak = \Aone \Atwo \cdots A_{(k-1)} A_{(k)}$ whose
$\Atwo$ is not proportional to $[\Aone \dg, C]$ for any $C$. 
For $D$ noncommutive with $A_{(k)}$, 
$\ak$ is the only operator 
whose commutator with the Hamiltonian gives $\bkp = \Aone \Atwo \cdots A_{(k-1)} [\boldA_{(k)},D] D\dg$: 
\eq{
\bmat{
\Aone & \Atwo & \cdots & A_{(k-1)} & A_{(k)} \\
      &       &        &            & D & D\dg \\ \hline
\Aone & \Atwo & \cdots & A_{(k-1)} & [A_{(k)},D] & D\dg
}
\label{eq:aaaad_k}
}
\footnote{
Here, we used the fact that $k\leq N/2$. 
Notice that for $k>N/2$, the list of $\{\boldA^k_j\}$ contributing to $\bki$ becomes more complicated.
For a detail discussion, see Sec.~VI A in Ref.~\cite{Chiba2024}.
}.
Thus, the coefficient of this $\ak$ is zero. 

\change{Next, we consider 
\eq{
\ak \propto 
\bmatc{
\hline
\Cone & \Cone\dg &          &        &         &        &         \\
     & \Ctwo  & \Athree & A_{(4)} & \cdots & A_{(k-1)} & A_{(k)} \\\hline
}{|ccccccc|}~,
}
where $\Athree$ is not proportional to 
$[\Ctwo^\dagger, C]$
for any $C$. 
For $D$ noncommutive with $A_{(k)}$, 
we analyze all the commutators that generate the following $\boldB^{k+1}$:
\eq{
\boldB^{k+1}=
\bmatc{
\hline
\Cone & \Cone\dg &        &         &        &           & &\\
      & \Ctwo  & \Athree & A_{(4)} & \cdots & A_{(k-1)} & A_{(k)}& \\
      &          &        &         &        &           & D & D\dg \\\hline
}{|cccccccc|}~.
}
We have the following two commutators:
\eq{
\bmatc{
\cline{1-7}
\Ic{\Cone} & \Cone\dg &        &         &        &      &  \cI{}   &\\
\Ic{}    & \Ctwo  & \Athree & A_{(4)} & \cdots & A_{(k-1)} & \cI{A_{(k)}} &\\\cline{1-7}
&&& &     &     & D & D\dg \\ \hline \\[-0.8em]
\cline{1-8}
\Ic{\Cone} & \Cone\dg &        &         &        &           & &\cI{}\\
\Ic{}     & \Ctwo  & \Athree & A_{(4)} & \cdots & A_{(k-1)} & A_{(k)}& \cI{}\\
\Ic{}      &          &        &         &        &           & D & \cI{D\dg}\\
\cline{1-8}
}{cccccccc}~
}
and
\eq{
\bmatc{
\cline{2-8}
& \Ic{\Ctwo}  & \Athree & A_{(4)} & \cdots & A_{(k-1)} & A_{(k)} & \cI{}\\
& \Ic{} &  &        &         &        &   D   &  \cI{D\dg}   \\
\cline{2-8}
\Cone & \Cone\dg && &     &     &  &  \\ \hline \\[-0.8em]
\cline{1-8}
\Ic{\Cone} & \Cone\dg &        &         &        &           & &\cI{}\\
\Ic{}     & \Ctwo  & \Athree & A_{(4)} & \cdots & A_{(k-1)} & A_{(k)}& \cI{}\\
\Ic{}      &          &        &         &        &           & D & \cI{D\dg}\\
\cline{1-8}
}{cccccccc}~.
}
Since $r_{\boldB^{k+1}_i}=0$, we have
\eq{
&q~\bmatc{
\hline
\Cone & \Cone\dg &          &        &         &        &         \\
     & \Ctwo  & \Athree & A_{(4)} & \cdots & A_{(k-1)} & A_{(k)} \\\hline
}{|ccccccc|}_{~i}\nx
&=\frac{c_{(1)}}{d}
q~\bmatc{
\hline
 \Ctwo  & \Athree & A_{(4)} & \cdots & A_{(k-1)} & A_{(k)} &\\
           &        &         &        &           & D & D\dg \\\hline
}{|ccccccc|}_{~i+1}\nx
&=0~, 
}
where we use the fact that we proved in the first part. }

A similar argument allows us to prove the general case by induction.
For example, we consider 
\eq{
\ak \propto 
\bmatc{
\hline
\Cone & \Cone\dg &          &        &         &        &           & \\
      & \Ctwo    & \Ctwo\dg &        &         &        &           & \\
      &          & \Cthree  & \Afour & A_{(5)} & \cdots & A_{(k-1)} & A_{(k)} \\\hline
}{|cccccccc|}~,
}
where $\Afour$ is not proportional to 
$[\Cthree^\dagger, C]$
for any $C$. 
Then, we have
\eq{
&q~\bmatc{
\hline
\Cone & \Cone\dg &          &        &         &        &           & \\
      & \Ctwo    & \Ctwo\dg &        &         &        &           & \\
      &          & \Cthree  & \Afour & A_{(5)} & \cdots & A_{(k-1)} & A_{(k)} \\\hline
}{|cccccccc|}_{~i}\nx
&\propto
q~\bmatc{
\hline
\Ctwo & \Ctwo\dg &        &         &        &           & &\\
      & \Cthree  & \Afour & A_{(5)} & \cdots & A_{(k-1)} & A_{(k)}& \\
      &          &        &         &        &           & \tilde{C}_{(1)} & \tilde{C}_{(1)}\dg \\\hline
}{|cccccccc|}_{~i+1}\nx
&\propto
q~\bmatc{
\hline
       \Cthree  & \Afour & A_{(5)} & \cdots & A_{(k-1)} & A_{(k)}& &\\
                &        &         &        &           & \tilde{C}_{(1)} & \tilde{C}_{(1)}\dg &\\
                &        &         &        &           &                 & \tilde{C}_{(2)}    & \tilde{C}_{(2)}\dg \\\hline
}{|cccccccc|}_{~i+2}\nx
&=0~, 
}
as discussed above. 
Therefore, we have proven that $q_\aki=0$ for any $i$ and $\ak$  
not of the form of Eq.~\eqref{eq:almost_dp_k}.
\end{proof}

\begin{prop}
Assume $Q$ to be a $k$-local conserved quantity. 
$q_\aki=0$ holds for any $i$ unless $\ak$ is a doubling operator, that is,
\eq{
\ak = 
\bmatc{
\hline
\Cone & \Cone\dg &          &        &       &          &\\
      & \Ctwo    & \Ctwo\dg &        &       &          &\\
      &          &          & \ddots &       &          &\\
      &          &          &        & \Ckmm & \Ckmm\dg &\\
      &          &          &        &       & \Ckm     &\Ckm\dg \\\hline
}{|ccccccc|}~.
\label{eq:dp_k}}
\label{prop:dp_k}
\end{prop}
\begin{proof}
The proof is similar to that of Proposition~\ref{prop:dp_3}. 
First, if $A_{(k)} = \ez$ and $\Ckm = \fz$, then we have 
\eq{
&q~\bmatc{
\hline
\Cone & \Cone\dg &          &        &       &          &\\
      & \Ctwo    & \Ctwo\dg &        &       &          &\\
      &          &          & \ddots &       &          &\\
      &          &          &        & \Ckmm & \Ckmm\dg &\\
      &          &          &        &       & \fz     &\ez \\\hline
}{|ccccccc|}_{~i}\nx
&\propto \cdots\nx
&\propto
q~\bmatc{
\hline
\fz   & \ez      &          &        \\
      & \fpp     & \fmm     &     \\   
      &          &          & \ddots \\\hline
}{|cccc|}_{~i+k-2} \nx
&=0~. \label{eq:not_doubling_k_1}
}
by a similar discussion in Eq.~\eqref{eq:not_doubling_3}. 

If $A_{(k)} = \ez$ and $\Ckm \neq \ez,\fz$ hold, 
we obtain the following:
\eq{
&q~\bmatc{
\hline
\Cone & \Cone\dg &          &        &       &          &\\
      & \Ctwo    & \Ctwo\dg &        &       &          &\\
      &          &          & \ddots &       &          &\\
      &          &          &        & \Ckmm & \Ckmm\dg &\\
      &          &          &        &       & \Ckm     &\ez \\\hline
}{|ccccccc|}_{~i}\nx
&\propto \cdots\nx
&\propto
q~\bmatc{
\hline
\Ckm  & \ez  &          &        \\
      & \Ckm     & \Ckm\dg     &     \\   
      &          &          & \ddots \\\hline
}{|cccc|}_{~i+k-2} \nx
&=0~.\label{eq:not_doubling_k_2}
}
Therefore, the proof for the case $A_{(k)}= \ez$ is complete. 
In the case of $A_{(k)} \neq \ez$, 
the coefficient can be shown to be zero 
by reducing the discussion to Eqs.~\eqref{eq:not_doubling_k_1} and~\eqref{eq:not_doubling_k_2}, 
as in the latter part of the proof of Proposition~\ref{prop:dp_3}. 
\end{proof}

\begin{prop}
Assume $Q$ to be a $k$-local conserved quantity. 
\eq{
&q~\bmatc{
\hline
\Cone & \Cone\dg &          &        &       &          &\\
      & \Ctwo    & \Ctwo\dg &        &       &          &\\
      &          &          & \ddots &       &          &\\
      &          &          &        & \Ckmm & \Ckmm\dg &\\
      &          &          &        &       & \Ckm     &\Ckm\dg \\\hline
}{|ccccccc|}_{~i} \nx
&= q_k c_{(1)} c_{(2)} \cdots c_{(k-2)} c_{(k-1)}~,
\label{eq:dp_factor_k}
}
by using some constant $q_k$. 
Here, we denote the coupling constant of the interaction term $C_{(n),j} C_{(n),j+1}\dg$ as $c_{(n)}$. 
\label{prop:dp_factor_k}
\end{prop}
\begin{proof}
Using the periodic boundary condition, 
we can show that the coefficients of any two doubling operators are proportional, as is the case with Proposition~\ref{prop:dp_factor_3}. 
In addition, 
Eq.~\eqref{eq:dp_factor_k} is consistent with the equations available from $\{ r_{\bki} = 0\}$.
\end{proof}

\subsection{Step 2 (analysis of $r_{{\bf B}^{k}_i} = 0$) for $3 \leq k \leq N/2$}\label{subsec44}
In this step, 
the only remaining degree of freedom $q_k$ is eliminated 
by using the fact that each $k$-local term $\bki$ of the commutator $[Q,H]$ has zero coefficient.

We introduce a symbol representing the common operator $\ez$ acting on successive sites as 
\eq{
&\em \ezo{k-2} \ep \nx
&:= \em \underbrace{\ez \ez \cdots \ez \ez}_{k-2} \ep \nx
&= \bmatc{
\hline
\em & \ep &     &        &     &     &     \\
    & \em & \ep &        &     &     &     \\
    &     &     & \ddots &     &     &     \\
    &     &     &        & \em & \ep &     \\
    &     &     &        &     & \em & \ep \\ \hline
}{|ccccccc|}~.
}
Using this symbol, we state the $k$-local terms $\boldB^k_j$'s we use. 
Unlike the case of $k=3$, we focus on multiple $\boldB^k_j$'s for $k > 3$. 
We eliminate the degree of freedom $q_k$ by linking the equations obtained from them. 
Specifically, 
for integer $m \in [0,k-3]$, 
we use $\boldB^k_{i-m}$, expressed as follows:
\eq{
\boldB^k_{i-m} &= \begin{cases}
\bmatc{
\hline
\ez & \ez & \fz &           &  \\
    &     & \em & \ezo{k-4} & \ep \\\hline
}{|ccccc|}_{~i} \qquad  (m=0) \\ \\
\bmatc{
\hline
\em & \ezo{m-1} & \ep &     &     &             & \\
    &           & \ez & \ez & \fz &             & \\
    &           &     &     & \em & \ezo{k-m-4} & \ep \\\hline
}{|ccccccc|}_{~i-m}  \\ \qquad \qquad \qquad \qquad \qquad \qquad \qquad 
(1\leq m \leq k-4) \\ \\
\bmatc{
\hline
\em & \ezo{k-4} & \ep &     &     \\
    &           & \ez & \ez & \fz \\\hline
}{|ccccc|}_{~i-k+3}  \quad  (m=k-3)~.
\end{cases}
\label{eq:multiple_outputs}
}
In the following, we enumerate the commutators contributing to each $\boldB^k_{i-m}$ 
and obtain the equations $\{q_{\aki}\}$ and $\{q_{\akmi}\}$.

First, for $m=0$ case, we take
\eq{
\boldB^k_i &= 
\bmatc{
\hline
\ez & \ez & \fz &           &  \\
    &     & \em & \ezo{k-4} & \ep \\\hline
}{|ccccc|}_{~i}
}
and analyze all the commutators that generate this $\bki$.
All the commutators that produce this $\bki$ are classified into the following three cases, as in Eq.~\eqref{eq:zzw_q_3}:
(i) commutators of $\boldA^{k-1}_j$ of $Q$ and interaction terms of $H$ that share a single site, 
(ii) those of $\boldA^k_j$ and magnetic field terms that share a single site, 
and (iii) those of $\boldA^k_j$ and interaction terms that share two sites.

There is only one contribution in case (i), represented as
\eq{
\bmat{
\cline{2-6}
&\Ic{\ez} & \ez & \fz &           &\cI{}\\
&\Ic{}    &     & \em & \ezo{k-5} &\cI{\ep}\\ \cline{2-6}
&         &     &     &           & \em     & \ep \\ \hline \\[-0.8em] \cline{2-7}
+&\Ic{\ez}& \ez & \fz &           &         & \cI{} \\
&\Ic{}    &     & \em & \ezo{k-5} & \ez     &\cI{\ep}\\ \cline{2-7}
}~.\label{eq:output_from_k-1}
}

Next, no pair of $\boldA^k_j$ and a magnetic field term may contribute as in the case (ii),
because, for any doubling operator $\boldA^k_i$, 
commutators with $\fz$ in all positions do not generate this $\bki$.

Finally, we consider commutators of $\boldA^k_i$ and interaction terms. 
We divide the cases according to the relative position of the interaction term to $\boldA^k_i$. 
If the interaction is located at the first and second sites from the left of $\boldA^k_i$, the commutator is described as follows:
\eq{
\bmat{
\hline
\Ic{C}   &C\dg &     &           &  \cI{}\\
\Ic{}    & \fz & \fz &           &  \cI{}\\
\Ic{}    &     & \em & \ezo{k-4} &\cI{\ep}\\ \cline{1-5}
D        &D\dg &     &           &        \\ \hline \hline \\[-0.8em] \hline
\Ic{\ez} & \ez & \fz &           & \cI{} \\
\Ic{}    &     & \em & \ezo{k-4} &\cI{\ep}\\ \hline
}~,\label{eq:cd_k_left}
}
up to constant factors. 
Here, $C$, $D$, and the constant multipliers in the commutator 
are the same as in Eq.~\eqref{eq:zzw_3_from_left}:
\eq{
\begin{array}{l}
(C,D,{\rm factor})\\
= (\ep,\em,+3/2),(\ep,\fm,-3/2),\\
\quad (\fp,\em,-3/2),(\fp,\fm,+3/2),(\fpp,\fmm,0),\\
\quad (\em,\ep,+3/2),(\em,\fp,-3/2),\\
\quad (\fm,\ep,-3/2),(\fm,\fp,+3/2),(\fmm,\fpp,0)~.
\end{array}\nonumber}
If the interaction is located at the second and third sites from the left of $\boldA^k_j$, 
we have the following relations up to constant factors:
\eq{
\bmat{
\hline
\Ic{\ez} &\ez  &     &           &  \cI{}\\
\Ic{}    &C    &C\dg &           &  \cI{}\\
\Ic{}    &     & \em & \ezo{k-4} &\cI{\ep}\\ \cline{1-5}
         &D    &D\dg &           &        \\ \hline \hline \\[-0.8em] \hline
\Ic{\ez} & \ez & \fz &           & \cI{} \\
\Ic{}    &     & \em & \ezo{k-4} &\cI{\ep}\\ \hline
}~,\label{eq:cd_k_right}
}
with
\eq{
\begin{array}{l}
(C,D,{\rm factor})\\
= (\ep,\em,0),(\ep,\fm,0),\\
\quad (\fp,\em,-1/3),(\fp,\fm,-1/3),(\fpp,\fmm,0),\\
\quad (\em,\ep,-1/6),(\em,\fp,-1/6),\\
\quad (\fm,\ep,-1/2),(\fm,\fp,+1/6),(\fmm,\fpp,+1/3)~.
\end{array}\nonumber}
In fact, the relative position of $\boldA^k_i$ and the interaction term cannot be other than these cases.
The reasons for this are as follows:
There is no doubling operator that can be written in the form $\ez \ez \cdots$. 
Thus, either the first or second site from the left must change by commuting with the interaction term. 

By considering all the contributions described 
in Eqs.~\eqref{eq:output_from_k-1}-\eqref{eq:cd_k_right}, 
we obtain the following equations:
\begin{widetext}
\eq{
0 &= r~\bmatc{
\hline
    \ez & \ez & \fz &           &  \\
    &     & \em & \ezo{k-4} & \ep \\\hline
}{|ccccc|}_{~i} \nx
&= 
+\cem q~\bmatc{
\hline
\ez & \ez & \fz &           &  \\
    &     & \em & \ezo{k-5} & \ep \\\hline
}{|ccccc|}_{~i} \nx
&\qquad \begin{alignedat}{6}
            + q_k e_{-1}^{k-3} 
            \Bigl( &\cfz (+3/2 &&\cep \cem -3/2 &&\cep \cfm &&-3/2 \cfp \cem &&+3/2 \cfp \cfm &&+\quad0 \cfpp \cfmm \\
                   &\quad +3/2 &&\cem \cep -3/2 &&\cem \cfp &&-3/2 \cfm \cep &&+3/2 \cfm \cfp &&+\quad0 \cfmm \cfpp) \\
              +&\cez (+\quad0&&\cep \cem +\quad0 &&\cep \cfm &&-1/3 \cfp \cem &&-1/3 \cfp \cfm &&+\quad0 \cfpp \cfmm \\
                &\quad  -1/6 &&\cem \cep -1/6 &&\cem \cfp &&-1/2 \cfm \cep &&+1/6 \cfm \cfp &&+1/3 \cfmm \cfpp)\Bigr)~.
\end{alignedat}\label{eq:condition_k_left}}

Next, for $1 \leq m \leq k-4$, we take
\eq{
\boldB^k_{i-m} &= 
\bmatc{
\hline
\em & \ezo{m-1} & \ep &     &     &             & \\
    &           & \ez & \ez & \fz &             & \\
    &           &     &     & \em & \ezo{k-m-4} & \ep \\\hline
}{|ccccccc|}_{~i-m}~.
}
All the commutators generating this $\boldB^k_{i-m}$ are
\eq{
\bmat{
\cline{3-9}
&         &\Ic{\em} & \ezo{m-2} & \ep &     &     &             & \cI{}\\
&         &\Ic{}    &           & \ez & \ez & \fz &             & \cI{}\\
&         &\Ic{}    &           &     &     & \em & \ezo{k-m-4} & \cI{\ep}  \\ \cline{3-9}
& \em     &\ep      &           &     &     &     &             &           \\ \hline \\[-0.8em] \cline{2-9}
-&\Ic{\em}& \ez     & \ezo{m-2} & \ep &     &     &             & \cI{}\\
& \Ic{}   &         &           & \ez & \ez & \fz &             & \cI{}\\
& \Ic{}   &         &           &     &     & \em & \ezo{k-m-4} &\cI{\ep} \\ \cline{2-9}
}\qquad
\bmat{
\cline{2-8}
&\Ic{\em} & \ezo{m-1} & \ep &     &     &             & \cI{}\\
&\Ic{}    &           & \ez & \ez & \fz &             & \cI{}\\
&\Ic{}    &           &     &     & \em & \ezo{k-m-5} & \cI{\ep}  \\ \cline{2-8}
&         &           &     &     &     &             & \em &\ep \\ \hline \\[-0.8em] \cline{2-9}
+&\Ic{\em}& \ezo{m-1} & \ep &     &     &             &     & \cI{}\\
& \Ic{}   &           & \ez & \ez & \fz &             &     & \cI{}\\
& \Ic{}   &           &     &     & \em & \ezo{k-m-5} & \ez & \cI{\ep} \\ \cline{2-9}
}~,
} 
\eq{
\bmat{
\hline
\Ic{\em} & \ezo{m-1} & \ep &     &     &             & \cI{}\\
\Ic{}    &           & C   &C\dg &     &             & \cI{}\\
\Ic{}    &           &     & \fz & \fz &             & \cI{}\\
\Ic{}    &           &     &     & \em & \ezo{k-m-4} & \cI{\ep}  \\ \hline
         &           & D   &D\dg &     &             &  \\ \hline \hline \\[-0.8em] \hline
\Ic{\em} & \ezo{m-1} & \ep &     &     &             & \cI{}\\
\Ic{}    &           & \ez & \ez & \fz &             & \cI{}\\
\Ic{}    &           &     &     & \em & \ezo{k-m-4} & \cI{\ep} \\ \hline
}\qquad
{\rm with }\,\begin{array}{l}
(C,D,{\rm factor})\\
= (\ep,\em,0),(\ep,\fm,0),\\
\quad (\fp,\em,-3),(\fp,\fm,-3),(\fpp,\fmm,0),\\
\quad (\em,\ep,+3/2),(\em,\fp,+3/2),\\
\quad (\fm,\ep,-3/2),(\fm,\fp,+9/2),(\fmm,\fpp,0)~,
\end{array}\label{eq:cd_middle_left}
}
\eq{
\bmat{
\hline
\Ic{\em} & \ezo{m-1} & \ep &     &     &             & \cI{}\\
\Ic{}    &           & \ez & \ez &     &             & \cI{}\\
\Ic{}    &           &     & C   &C\dg &             & \cI{}\\
\Ic{}    &           &     &     & \em & \ezo{k-m-4} & \cI{\ep}  \\ \hline
         &           &     & D   &D\dg &             &  \\ \hline \hline \\[-0.8em] \hline
\Ic{\em} & \ezo{m-1} & \ep &     &     &             & \cI{}\\
\Ic{}    &           & \ez & \ez & \fz &             & \cI{}\\
\Ic{}    &           &     &     & \em & \ezo{k-m-4} & \cI{\ep} \\ \hline
}\qquad
{\rm with }\,\begin{array}{l}
(C,D,{\rm factor})\\
= (\ep,\em,0),(\ep,\fm,0),\\
\quad (\fp,\em,-1/3),(\fp,\fm,-1/3),(\fpp,\fmm,0),\\
\quad (\em,\ep,-1/6),(\em,\fp,-1/6),\\
\quad (\fm,\ep,-1/2),(\fm,\fp,+1/6),(\fmm,\fpp,+1/3)~,
\end{array}\label{eq:cd_middle_right}
}
which read
\eq{
0 &= r~\bmatc{
\hline
\em & \ezo{m-1} & \ep &     &     &             & \\
    &           & \ez & \ez & \fz &             & \\
    &           &     &     & \em & \ezo{k-m-4} & \ep \\\hline
}{|ccccccc|}_{~i-m} \nx
&= -\cem q~\bmatc{
\hline
\Ic{\em} & \ezo{m-2} & \ep &     &     &             & \cI{}\\
\Ic{}    &           & \ez & \ez & \fz &             & \cI{}\\
\Ic{}    &           &     &     & \em & \ezo{k-m-4} & \cI{\ep}  \\ \hline}{ccccccc}_{~i-m}
+\cem q~\bmatc{
\hline
\Ic{\em} & \ezo{m-1} & \ep &     &     &             & \cI{}\\
\Ic{}    &           & \ez & \ez & \fz &             & \cI{}\\
\Ic{}    &           &     &     & \em & \ezo{k-m-5} & \cI{\ep}  \\ \hline}{ccccccc}_{~i-m+1}
 \nx
&\qquad \begin{alignedat}{7}
            + q_k e_{-1}^{k-3} 
            \Bigl( &\cfz (+\quad0 &&\cep \cem +\quad0   &&\cep \cfm -\quad3 &&\cfp \cem -\quad3 &&\cfp \cfm +\quad0 &&\cfpp \cfmm \\
              &\quad  +3/2 &&\cem \cep +3/2 &&\cem \cfp -3/2 &&\cfm \cep +9/2 &&\cfm \cfp +\quad0 &&\cfmm \cfpp) \\
              +&\cez (+\quad0   &&\cep \cem +\quad0   &&\cep \cfm -1/3 &&\cfp \cem -1/3 &&\cfp \cfm +\quad0 &&\cfpp \cfmm \\
              &\quad -1/6 &&\cem \cep -1/6 &&\cem \cfp -1/2 &&\cfm \cep +1/6 &&\cfm \cfp +1/3 &&\cfmm \cfpp)\Bigr)~.
\end{alignedat}\label{eq:condition_k_middle}}

Finally, for $m = k-3$, we consider
\eq{
\boldB^k_{i-k+3} &= 
\bmatc{
\hline
\em & \ezo{k-4} & \ep &     &     \\
    &           & \ez & \ez & \fz \\\hline
}{|ccccc|}_{~i-k+3}~.
}
All the commutators generating this $\boldB^k_{i-k+3}$ are
\eq{
\bmat{
\cline{3-7}
&         &\Ic{\em} & \ezo{k-4} & \ep &     & \cI{} \\
&         &\Ic{}    &           & \ez & \ez & \cI{\fz}    \\ \cline{3-7}
&\em      & \ep     &           &           &        \\ \hline \\[-0.8em] \cline{2-7}
-&\Ic{\em}& \ez     & \ezo{k-4} & \ep &     & \cI{} \\
&\Ic{}    &         &           & \ez & \ez & \cI{\fz}\\ \cline{2-7}
}
~,
}
\eq{
\bmat{
\hline
\Ic{\em} & \ezo{k-4} & \ep &     & \cI{} \\
\Ic{}    &           & C   &C\dg & \cI{} \\
\Ic{}    &           &     & \fz & \cI{\fz} \\ \hline
         &           & D   &D\dg &        \\ \hline \hline \\[-0.8em] \hline
\Ic{\em} & \ezo{k-4} & \ep &     & \cI{} \\
\Ic{}    &           & \ez & \ez & \cI{\fz} \\ \hline
}\qquad
\begin{array}{l}
(C,D,{\rm factor})\\
= (\ep,\em,0),(\ep,\fm,0),\\
\quad (\fp,\em,-3),(\fp,\fm,-3),(\fpp,\fmm,0),\\
\quad (\em,\ep,+3/2),(\em,\fp,+3/2),\\
\quad (\fm,\ep,-3/2),(\fm,\fp,+9/2),(\fmm,\fpp,0)~,
\end{array}
}
\eq{
\bmat{
\hline
\Ic{\em} & \ezo{k-4} & \ep &     & \cI{} \\
\Ic{}    &           & \ez & \ez & \cI{} \\
\Ic{}    &           &     & C   & \cI{C\dg} \\ \hline
         &           &     & D   & D\dg       \\ \hline \hline \\[-0.8em] \hline
\Ic{\em} & \ezo{k-4} & \ep &     & \cI{} \\
\Ic{}    &           & \ez & \ez & \cI{\fz} \\ \hline
}\qquad
\begin{array}{l}
(C,D,{\rm factor})\\
= (\ep,\em,-1/6),(\ep,\fm,-1/2),\\
\quad (\fp,\em,-1/2),(\fp,\fm,-1/6),(\fpp,\fmm,+1/3),\\
\quad (\em,\ep,-1/6),(\em,\fp,-1/2),\\
\quad (\fm,\ep,-1/2),(\fm,\fp,-1/6),(\fmm,\fpp,+1/3)~.
\end{array}
}
It should be noted that the factors in the second and last commutators
correspond to Eq.~\eqref{eq:cd_middle_left} and Eq.~\eqref{eq:zzw_3_from_right}, respectively. 
These three forms give the following equation:
\eq{
0 &= r~\bmatc{
\hline
\em & \ezo{k-4} & \ep &     &     \\
    &           & \ez & \ez & \fz \\\hline
}{|ccccc|}_{~i-k+3} \nx
&= -\cem q~\bmatc{
\hline
\em & \ezo{k-5} & \ep &     &  \\
    &           & \ez & \ez & \fz  \\ \hline}{|ccccc|}_{~i-k+4}
 \nx
&\qquad \begin{alignedat}{6}
            + q_k e_{-1}^{k-3} 
            \Bigl( &\cfz (+\quad0   &&\cep \cem +\quad0 &&\cep \cfm -\quad3 &&\cfp \cem -\quad3 &&\cfp \cfm +\quad0 &&\cfpp \cfmm \\
              & \quad  +3/2 &&\cem \cep +3/2 &&\cem \cfp -3/2 &&\cfm \cep +9/2 &&\cfm \cfp +\quad0 &&\cfmm \cfpp) \\
              +&\cez (-1/6 &&\cep \cem -1/2 &&\cep \cfm -1/2 &&\cfp \cem -1/6 &&\cfp \cfm +1/3 &&\cfpp \cfmm \\
              & \quad   -1/6 &&\cem \cep -1/2 &&\cem \cfp -1/2 &&\cfm \cep -1/6 &&\cfm \cfp +1/3 &&\cfmm \cfpp)\Bigr)~.
\end{alignedat}\label{eq:condition_k_right}}
By adding together Eqs.~\eqref{eq:condition_k_left}, \eqref{eq:condition_k_middle}, and~\eqref{eq:condition_k_right}, 
$\{ q_{\boldA^{k-1}_j}\}$ can be eliminated, 
and we obtain an equation with only $\{ q_{\boldA^{k-1}_j}\}$ as follows:
\eq{
0 = q_k (k-1) e_{-1}^{k-3} \Bigl( \cfz &( +3/2 \cep \cem -3/2 \cep \cfm -3/2 \cfp \cem +3/2 \cfp \cfm ) \nx
                             +\cez &( -1/6 \cep \cem -1/2 \cep \cfm -1/2 \cfp \cem -1/6 \cfp \cfm + 1/3 \cfpp \cfmm) \Bigr)~,
\label{eq:necessary_condition_k}}
which is equivalent to Eq.~\eqref{eq:condition_3} under the assumption $\cem\neq0$. 
Therefore, 
together with exception handling in Appendix~\ref{sec:singular_case} and Sec.~\ref{sec:necessary}, 
the BLBQ model~\eqref{eq:H_BLBQ} has shown to have no $k$-local conserved quantity ($3\leq k\leq N/2$) 
except for known integrable systems. 
\end{widetext}

\subsection{$k=2$ case}\label{subsec45}
The argument in Step 1 holds exactly as in the general $k$ case (whereas the argument in Step 2 does not). 
As a result, 
the coefficients of the terms in the 2-local conserved quantities are zero 
except for the doubling operators $CC\dg$. 
Furthermore, the coefficients of the doubling operators can be expressed as
\eq{
q_{(C C\dg)_i} = q_{k=2} c~, 
}
by using a certain constant $q_{k=2}$. 
This means that 2-local operators of any 2-local conserved quantity are identical to the Hamiltonian. 
In other words, all 2-local conserved quantities are written as the sum of Hamiltonian and 1-local conserved quantities.
Therefore, there are no 2-local conserved quantities independent of the Hamiltonian. 

\subsection{$k=1$ case}\label{subsec46}
First, we consider the coefficients $q_{E_{m,i}}$. 
All the commutators generating $\boldB^2_i = (E_{+1} E_{-1})_i$ are as follows:
\eq{
\bmat{
 &\ez &     \\
 &\ep & \em \\ \hline
+&\ep & \em
}\qquad
\bmat{
 &    & \ez \\
 &\ep & \em \\ \hline
-&\ep & \em
}~,\label{eq:condition_absence_of_E0}
}
which read
\eq{
\cep q_{E_{0,i}} - \cep q_{E_{0,i+1}}=0
\quad\mbox{for all } i~.\label{eq:dummy_qe0i}}
Similarly, we consider all the commutators generating $\boldB^2_i = (E_{0} E_{\pm1})_i$ as
\eq{
\bmat{
 &E_{\pm1} &     \\
 &E_{\mp1} & E_{\pm1} \\ \hline
\pm&\ez & E_{\pm1}
}\qquad
\bmat{
 &    & E_{\pm1} \\
 &\ez & \ez \\ \hline
\mp&\ez & E_{\pm1}
}~,\label{eq:condition_absence_of_E+}
}
which read
\eq{
\cep q_{E_{\pm1,i}} = \cez q_{E_{\pm1,i+1}}
\quad\mbox{for all } i~.\label{eq:dummy_qe+i}
}
By substituting Eq.~\eqref{eq:couple_const_BLBQ} for these equations, 
we obtain 
\eq{
q_{E_{0,i}} &= q_\ez 
~,\nx
q_{E_{\pm1,i}} &= q_{E_{\pm1}} \quad \mathrm{independent \, of \,}i~.
}
When $D=0$, these equations correspond to the following conserved quantities:
\eq{
&\sumi E_{0,i} \qquad \left(= \sumi S^z_i\right) ~,\label{eq:total_z}\\
&\sumi E_{\pm 1,i} \qquad \left(= \frac{1}{\sqrt{2}} \sumi S^x_i \pm \frac{\im}{\sqrt{2}} \sumi S^y_i \right)~,\label{eq:total_pm}
}
that is, the total uniform magnetization. 
Meanwhile, 
if $D\neq0$, the first one, i.e., 
the total magnetization in the $z$ direction~\eqref{eq:total_z}, is conserved, 
but those in other directions~\eqref{eq:total_pm} are not. 
In fact, $q_{E_{\pm1}}=0$ when $D\neq0$ holds. 
This is shown by considering a $1$-local output $F_{\pm 1,i}$, 
which is generated only by 
\eq{
\bmat{
  & E_{\pm1} \\
  & \fz \\ \hline
\mp3 & F_{\pm1}
}~.\label{eq:no_e+}
}
Now that we have exhausted all independent conserved quantities 
for which $q_{E_{m,i}}\neq0$ for some $i$, 
we can assume $q_{E_{m,i}}=0$ for all $i$ in the following. 

Next, 
we show that there are no other $1$-local conserved quantities in the BLBQ model 
except for the case with $J_1\in\{0,J_2\}$. 
We consider all the commutators generating $\boldB^2_i = (F_{+1}E_{-1})_i$ as
\eq{
\bmat{
 &\fz &     \\
 &\ep & \em \\ \hline
3&\fp & \em
}\qquad
\bmat{
 &    & \fz \\
 &\fp & \fm \\ \hline
-3&\fp & \em
}~,
}
which read
\eq{
\cep q_{F_{0,i}} = \cfp q_{F_{0,i+1}}
\quad\mbox{for all } i~.
\label{eq:dummy_qf0i}
}
Here, if there exists a $1$-local conserved quantity with nonzero coefficient of $F_{0,i}$, 
then 
\eq{|\cep| = |\cfp|\label{eq:dummy_cepcfp}} 
must be satisfied. 
This is because the following holds:
\eq{\label{eq:abs_because}
\cep^N q_{F_{0,i}}
&= \cep^{N-1}\cfp  q_{F_{0,i+1}} \nx
&= \cep^{N-2}\cfp^2 q_{F_{0,i+2}} \nx
&= \cdots \nx
&= \cfp^N q_{F_{0,i}}~,
}
by the periodic boundary condition. 
Under Eq.~\eqref{eq:couple_const_BLBQ},
Eq.~\eqref{eq:dummy_cepcfp} implies
\eq{
J_1 - J_2/2 = \pm J_2/2~, 
\label{eq:condition_1-local}}
which is equivalent to $J_1\in\{0,J_2\}$. 
Thus, if $J_1\notin\{0,J_2\}$, 
$q_{F_{0,i}} = 0$ holds for all $i$. 

Similarly, we consider all the commutators generating $\boldB^2_i=F_{0,i}E_{\pm1,i+1}$ as
\eq{
\bmat{
 &F_{\pm1} &     \\
 &E_{\mp1} & E_{\pm1} \\ \hline
\pm&\fz & E_{\pm1}
}\qquad
\bmat{
 &    & F_{\pm1} \\
 &\fz & \fz \\ \hline
\mp3&\fz & E_{\pm1}
}~,
}
which read
\eq{
e_{\mp1} q_{F_{\pm1,i}} = 3 \cfz q_{F_{\pm1,i+1}}~.\label{eq:dummy_qfpmi}
}
Furthermore, we consider all the commutators generating $\boldB^2_i=F_{\pm1,i}E_{\pm1,i+1}$ as
\eq{
\bmat{
 &F_{\pm2}&     \\
 &E_{\mp1} & E_{\pm1} \\ \hline
\pm&F_{\pm1} & E_{\pm1}
}\qquad
\bmat{
 &    & F_{\pm2}\\
 &F_{\pm1} & F_{\mp1} \\ \hline
\mp&F_{\pm1} & E_{\pm1}
}~,
}
which read
\eq{
e_{\mp1} q_{F_{\pm2,i}} = f_{\pm1} q_{F_{\mp2,i+1}}~.\label{eq:dummy_qfppmmi}
}
By similar arguments to Eq.~\eqref{eq:abs_because}, Eqs.~\eqref{eq:dummy_qfpmi} and~\eqref{eq:dummy_qfppmmi} yield $q_{F_{\pm1,i}} =q_{F_{\pm2,i}}= 0$ and $q_{F_{+2,i}} = 0$, respectively, under the assumption $J_1 \notin \{0,J_2\}$. 
As a result, 
we show that the all of 1-local conserved quantities of the BLBQ model are Eq.~\eqref{eq:total_z} [and Eq.~\eqref{eq:total_pm} for $D=0$], 
except for the case with $J_1\in\{0,J_2\}$. 

Finally, we consider the model with $J_1\in\{0,J_2\}$. 
By Eqs.~\eqref{eq:dummy_qf0i}, \eqref{eq:dummy_qfpmi}, and~\eqref{eq:dummy_qfppmmi}, 
we have the following candidate for additional $1$-local conserved quantities: 
\eq{
&\sumi F_{0,i}\ \left(=3\sum_{i=1}^N (S_{i}^z)^2-2N\right) ~,\label{eq:total_f0}\\
&\sumi F_{\pm1,i}\ \left(= \sum_{i=1}^N\frac{\{S^z_i,S^x_i\}}{\sqrt{2}}
\pm \im\sum_{i=1}^N \frac{\{S^z_i,S^y_i\}}{\sqrt{2}}\right) ~,\label{eq:total_fpm}\\
&\sumi F_{\pm 2,i}\ \left(= \sum_{i=1}^N\frac{(S^x_i-S^y_i)^2}{2}
\pm \im\sum_{i=1}^N \frac{\{S^x_i,S^y_i\}}{2} \right)~,\label{eq:total_fppmm}
}
for $J_1=J_2$, and
\eq{
&\sumi (-1)^i F_{0,i}
~,\label{eq:stag_f0}\\
&\sumi (-1)^i F_{\pm1,i}
~,\label{eq:stag_fpm}\\
&\sumi (-1)^i F_{\pm 2,i}
~,\label{eq:stag_fppmm}
}
for $J_1=0$ and even $N$. 
We note that there is no additional $1$-local conserved quantity for $J_1=0$ if $N$ is odd 
by Eqs.~\eqref{eq:dummy_qf0i}, \eqref{eq:dummy_qfpmi} and~\eqref{eq:dummy_qfppmmi}. 

In the case of $J_1=J_2$ or both $J_1=0$ and even $N$,
all of them are actual conserved quantities if $D=0$. 
The first one [Eq.~\eqref{eq:total_f0} for $J_1=J_2$ and Eq.~\eqref{eq:stag_f0} for $J_1=0$] 
and the third one [Eq.~\eqref{eq:total_fppmm} for $J_1=J_2$ and Eq.~\eqref{eq:stag_fppmm} for $J_1=0$ and even $N$] 
are also conserved even for $D\neq0$. 
On the other hand, the coefficients of $F_{\pm1,i}$ vanish for $D\neq0$ by a similar argument in Eq.~\eqref{eq:no_e+}. 

Now, we summarize the above conclusions for the case of a nonintegrable system ($D\neq0$, $J_1=0$, and even $N$). 
In this case, there are three 1-local conserved quantities other than the total magnetization, 
which are in Eqs.~\eqref{eq:stag_f0} and~\eqref{eq:stag_fppmm}. 
Rewriting these quantities in spin basis and 
further expressing those as three independent Hermitian observables, 
we obtain
\eq{
&\sumi (-1)^i(S_i^z)^2~,\nx 
&\sumi (-1)^i((S_i^x)^2-(S_i^y)^2)~,\nx 
&\sumi (-1)^i(S_i^x S_i^y+S_i^y S_i^x)~.
}
This completes a proof of Theorem~\ref{thm:1and2}.

\section{
Nonintegrability of systems with anisotropic interactions
}\label{sec:necessary}
In Step~2 of the proof in Sec.~\ref{sec:proof}, 
we only considered one family of outputs $\boldB^k_{i-m}$ [Eq.~\eqref{eq:multiple_outputs}]. 
By using other outputs, we can obtain other equation for the coefficients. 
For example, we take a family of outputs $\boldB^k$ as the following 
type of operators:
\eq{
\bmatc{
\hline
 \ddots & & & & & \\
 \fpp & \fmm &     &     &             & \\
           & E_{+1,i} & \ep & \fmm &             & \\
           &     &     & \tilde{A} & \tilde{A}^{\dg} &  \\
   &   &   &   &   \ddots & \\
    \hline
}{|cccccc|}~.
\label{eq:multiple_outputs_b}
}
The fact that all coefficients of those are zero yields the following 
equation:
\begin{widetext}
    \eq{
    q_k\Bigl((\cep-\cfp)(\cfpp(\cez+3\cfz+6h)-(\cep^2+6\cep\cfp+\cfp^2))-\cfpp^2(\cem-\cfm) \Bigr) =0~.
    \label{eq:coeff_codition_3}}
\end{widetext}
Substituting Eq.~\eqref{eq:couple_const_BLBQ}, 
Eq.~\eqref{eq:coeff_codition_3} reduces to
\eq{-q_k(J_1-J_2)(J_1^2+J_1 J_2-6h J_2)=0~,}
which gives the rest of the proof of nonintegrability, i.e., 
for the model with $J_1\in\{0,-J_2\}$ and $D\neq0$. 

Up to this point, we restricted the discussion to the BLBQ model~\eqref{eq:H_BLBQ}. 
However, the exactly same proof is valid as long as 
all coupling constants are nonzero. 
As a more general model to which our discussion is applicable, 
we introduce the following model:
\begin{align}
    H&=\sum_{j=1}^{N}\Bigl(\sum_{m\in\{0,\pm 1\}}e_{m}E_{m,j}E_{-m,j+1}\notag\\
    &\hspace{40pt}+\sum_{m\in\{0,\pm 1,\pm 2\}}f_{m}F_{m,j}F_{-m,j+1}+hF_{0,j}\Bigr)~,
    \label{eq:generalizedH}
\end{align}
where $h\in\mathbb{R}$ and $e_m, f_m \in \mathbb{C}\setminus\{0\}$.
Here, we impose Hermiticity: 
Except for the restriction of coefficients to nonzero values, this extended model~\eqref{eq:generalizedH} can be characterized as a general model having the following three symmetries: the translation~$T$, the spin rotation around $z$ axis by arbitrary degree~$R_{z}(\theta)$, and the time reversal~$\Theta$. Each of them transforms the spin operators as
\begin{align}
    T\vec{S}_{j}T^{-1}&=\vec{S}_{j+1}\nx
    R_{z}(\theta)S_{j}^zR_{z}^{-1}(\theta)&=S_{j}^z\nx
    R_{z}(\theta)S_{j}^xR_{z}^{-1}(\theta)&=\cos\theta S_{j}^x+\sin\theta S_{j}^y\nx
    R_{z}(\theta)S_{j}^yR_{z}^{-1}(\theta)&=\cos\theta S_{j}^y-\sin\theta S_{j}^x\nx
    \Theta\vec{S}_{j}\Theta^{-1}&=-\vec{S}_{j}~.
\end{align}
Note that, through the time reversal $\Theta$, $E_{m}$ and $F_{m}$ change their sign of $m$, as in $\Theta E_{m} \Theta^{-1}=- E_{-m}$ and $\Theta F_{m} \Theta^{-1}= F_{-m}$, 
because of the antiunitarity of $\Theta$.

This model includes not only the BLBQ model~\eqref{eq:H_BLBQ}, 
but also models with anisotropic interactions such as 
the Fateev--Zamolodchikov model~\cite{Zamolodchikov1980}:
\begin{align}
    H_{\mathrm{FZ}}(\Delta)&=\sum_{j=1}^{N}\Bigl(\vec{S}_{j}\cdot\vec{S}_{j+1}-(\vec{S}_{j}\cdot\vec{S}_{j+1})^2\notag\\
     +&2(\Delta^2-1)\bigl(S_{j}^zS_{j+1}^z+2(S_{j}^z)^2-(S_{j}^z)^2(S_{j+1}^z)^2\bigr)\notag\\
    &\quad -2(\Delta-1)\{S_{j}^xS_{j+1}^x+S_{j}^yS_{j+1}^y,S_{j}^zS_{j+1}^z\}\Bigr)~.
    \label{eq:H_ZF}
\end{align}
This model is known to be integrable for any $\Delta\in\mathbb{R}$, 
and 
the isotropic case, which is given by $\Delta =1$, corresponds to the Takhtajan--Babujian model [$J_1=-J_2,\ D=0$ in Eq.~\eqref{eq:H_BLBQ}]. 
The Hamiltonian~\eqref{eq:H_ZF} 
corresponds to the extended model~\eqref{eq:generalizedH} with the following parameters:
\begin{align}
    &e_{0}=-\frac{1}{2}+2\Delta^2,\quad e_{\pm 1}=\frac{1}{2}+\Delta~,\notag\\
    &f_{0}=\frac{1}{18}-\frac{2\Delta^2}{9},\quad f_{\pm 1}=\frac{1}{2}-\Delta,\quad f_{\pm 2}=-1,\notag\\
    &h=-\frac{4}{9}+\frac{4\Delta^2}{9}~.
    \label{eq:H_ZF_parameters}
\end{align}
\change{This model satisfies $e_m, f_m \in \mathbb{R}$, 
which is equivalent to imposing an additional symmetry under the spin-$z$ flip $R_{x}(\pi)$, 
which transform the spin operators as
\begin{align}
    R_{x}(\pi)S_{j}^xR_{x}^{-1}(\pi)&=S_{j}^x~,\\
    R_{x}(\pi)S_{j}^yR_{x}^{-1}(\pi)&=-S_{j}^y~,\\
    R_{x}(\pi)S_{j}^zR_{x}^{-1}(\pi)&=-S_{j}^z~.
\end{align}
Some models with real $e_m$ and $f_m$ possess exactly solvable ground states~\cite{Klumper1993,Bartel2003}. 
In particular, it can be regarded as a generalization of the AKLT model.}

\change{By continuing the analysis of Step~2, we obtain the following result 
for the extended model with spin-$z$ flip symmetry.
\begin{thm}\label{thm:extended}
    In the extended model~\eqref{eq:generalizedH} with nonzero real coupling constants,  
    there are no 
    $k$-local conserved quantities with $3 \leq k \leq N/2$, 
    except for the following specific cases.
    The tuples below represent $(e_0,e_+,f_0,f_+,h)$, 
    where the proportionality constant is fixed by setting $\cfpp=1$:
\begin{enumerate}
    \item\label{enum:extended_Case1} $(\frac{\epsilon_1}{2},\frac{\epsilon_2}{2},\frac{\epsilon_1+2\epsilon_3}{18},\frac{\epsilon_2}{2},h_0)$
    \item $(-\frac{1}{2},-\frac{\epsilon_1}{2},\frac{1}{6},\frac{\epsilon_1}{2},0)$
    \item $(-1,\frac{\epsilon_1}{2},\frac{1}{3},-\frac{\epsilon_1}{2},0)$
    \item $(-\frac{\varphi+1}{2},\frac{\sqrt{\varphi}\epsilon_1}{2},\frac{1+3\varphi}{18},-\frac{\sqrt{\varphi}\epsilon_1}{2},\frac{\varphi-3}{18\varphi})$
    \item $(-\frac{1}{2},-\frac{\sqrt{2}\epsilon_1}{4},\frac{5}{18},\frac{\sqrt{2}\epsilon_1}{4},\frac{1}{36})$
    \item\label{enum:extended_Case6} $(\frac{1-4\Delta^2}{2},\frac{\epsilon_1-2\Delta}{2},\frac{4\Delta^2-1}{18},\frac{\epsilon_1+2\Delta}{2},\frac{4(1-\Delta^2)}{9})$
\end{enumerate}
    Here, $\epsilon_i\in\{\pm1\}$, $h_0\in\mathbb{R}$, 
    and $\Delta\in\mathbb{R}\setminus\{\pm\frac{1}{2}\}$ are arbitrary parameters, and $\varphi=(1+\sqrt{5})/2$ denotes the golden ratio.
\end{thm}
\begin{proof}[Proof (outline)]
    For example, from Eq.~\eqref{eq:coeff_codition_3}, we find that either $\cep=\cfp$ or
    \begin{equation}
        \cfpp(\cez+3\cfz+6h)-(\cep^2+6\cep\cfp+\cfp^2)-\cfpp^2=0~.
    \end{equation}
    When $\cep=\cfp$, Eq.~\eqref{eq:necessary_condition_k} implies $2|\cep|=|\cfpp|$. 
    Setting $\cfpp=1$ yields $\cep=\frac{\epsilon_2}{2}$.
    In general, by examining outputs other than those in Eqs.~\eqref{eq:multiple_outputs} and \eqref{eq:multiple_outputs_b}, 
    we obtain additional polynomial relations among the coefficients.
    Solving these equations, 
    we find that the only admissible solutions are those listed above.
\end{proof}
For explicit lists of the resulting polynomials and for the cases without spin-$z$ symmetry, see Appendix~\ref{sec:list_of_eq}.}

\change{From Eq.~\eqref{eq:H_ZF_parameters}, it follows that 
case~\ref{enum:extended_Case6} with $\epsilon_1=-1$ corresponds to the Fateev--Zamolodchikov model, which is solved by the Bethe ansatz~\cite{Mezincescu1990}.
In fact, all these cases~\ref{enum:extended_Case1}--\ref{enum:extended_Case6} are known~\footnote{
Precisely speaking, the first solution appears in Refs.~\cite{idzumi1994,BibikovNuramatov2016} only when $h=\frac{2\epsilon_1-2\epsilon_3}{9}$.
However, since the corresponding $R$-matrix commutes with $F_{0}\otimes I+I\otimes F_0$, adding a term proportional to $\sum_{i=1}^NF_{0,i}$ does not affect integrability.
} 
to be derived from the Yang--Baxter equation~\cite{idzumi1994,BibikovNuramatov2016}, 
which implies integrablity of these models~\cite{Luscher1976}.
Therefore, Theorem~\ref{thm:extended} demonstrates that, 
within the extended model~\eqref{eq:generalizedH} with nonzero real coupling constants, 
all models are nonintegrable except for these known integrable cases~\ref{enum:extended_Case1}--\ref{enum:extended_Case6}.}

\section{Conclusion and Discussion}\label{sec:conclusion}
We have proved that except for the known integrable systems, 
all systems in the BLBQ model with an anisotropic field are nonintegrable, 
by showing the absence of $k$-local conserved quantities with $3\leq k\leq N/2$. 
In particular, by proving that the AKLT model is nonintegrable, 
the existence of nonthermal energy eigenstates in that model is not caused by integrability.
We have also demonstrated the nonintegrability of systems with anisotropic interactions beyond the BLBQ model.
Our results provide the first proof of nonintegrability in spin-$1$ systems, 
which was considered to be a challenging extension of nonintegrability proofs~\cite{Shiraishi2019}. 
        
The operator basis~\eqref{eq:basis} we used is based on spherical harmonic functions. 
Such an operator basis is similarly constructed for the $S\geq 3/2$ spins~\cite{Madore1992}. 
Although some commutators of two basis elements yield linear combinations of multiple basis elements in higher spins, 
we believe that our proof of nonintegrability in spin-$1$ systems 
provides useful insights into the possibility of that in spin $S\geq 3/2$ systems. 

\textit{Note added.}
During the completion of this manuscript, we became aware of an independent related work by \mbox{Park} and \mbox{Lee}~\cite{Park2025BLBQ}.

\begin{acknowledgments}
The authors thank Naoto Shiraishi for fruitful discussions.
We also appreciate Hosho Katsura and Nicholas O'Dea for their feedback on the manuscript.
A.H. was supported by KAKENHI
Grant No. JP25KJ0833 from the Japan Society for the
Promotion of Science (JSPS) and 
FoPM, a WINGS Program, the University of Tokyo. 
A.H. also acknowledges support from JSR Fellowship, 
the University of Tokyo.
M.Y. was supported by KAKENHI Grant No. JP25KJ0815 from JSPS and FMSP, a WINGS Program, The University of Tokyo.
Y.C. was supported by the Special Postdoctoral Researchers Program at RIKEN and JST ERATO Grant No.~JPMJER2302, Japan.
\end{acknowledgments}

\appendix

\section{For $J_2=0$ and $J_2=2J_1$}\label{sec:singular_case}
In the main text, we omitted the case $J_2=0,2J_1$, 
where some coupling constants are zero. 
In this section, 
we discuss the local conserved quantities in these two parameter regions 
and complete the proof of our results.

Before the proof, we need to clarify some terminology. 
So far, we have used the two words, 
elements of the operator basis and interaction terms in the same sense, 
but in the case of $J_2=0,2J_1$ the two are different concepts. 
That is, among the elements of the operator basis, 
we refer to only those with nonzero coupling constants as interaction terms.
In addition, among operators of the form in Eq.~\eqref{eq:dp_k}, 
we call the doubling operator only those 
for which all $C_{(n)}$ are interaction terms. 

\subsection{Step 1 (analysis of $r_{{\bf B}^{k+1}_i} = 0$)}\label{subsec_appx_step1}
Let $k$ be an integer between $2$ and $N/2$. 
First, we show that the $k$-local term of a $k$-local conserved quantity is a doubling operator 
in two steps, as in the main text. 
\begin{prop}\label{prop:almost_dp_singular}
Assume $Q$ to be a $k$-local conserved quantity of the BLBQ Hamiltonian~\eqref{eq:H_BLBQ} with $J_2=0$ or $J_2=2J_1$.
$q_\aki=0$ holds for any $i$ unless 
there are interaction terms $\{C_{(n)}\}$ and $A_{(k)}$ satisfying
\eq{
\ak \propto 
\bmatc{
\hline
\Cone & \Cone\dg &          &        &       &          &\\
      & \Ctwo    & \Ctwo\dg &        &       &          &\\
      &          &          & \ddots &       &          &\\
      &          &          &        & \Ckmm & \Ckmm\dg &\\
      &          &          &        &       & \Ckm     & A_{(k)}\\\hline
}{|ccccccc|}~.
\label{eq:almost_dp_singular}
}
\end{prop}
\begin{proof}
The proof is similar to that of Proposition~\ref{prop:almost_dp_k}. 
We note that in both cases $J_2=0$ and $J_2=2J_1$, 
for any element $A$ of the operator basis 
some interaction term does not commute with $A$. 

First, we consider the case where $\ak = \Aone \Atwo \cdots A_{(k-1)} A_{(k)}$ and 
$\Atwo$ is not proportional to $[\Aone \dg, C]$ for any $C$. 
For some interaction term $D$ that does not commute with $A_{(k)}$, 
$\ak$ is the only operator 
whose commutator with the Hamiltonian gives $\bkp = \Aone \Atwo \cdots A_{(k-1)} [\boldA_{(k)},D] D\dg$: 
\eq{
\bmat{
\Aone & \Atwo & \cdots & A_{(k-1)} & A_{(k)} \\
      &       &        &            & D & D\dg \\ \hline
\Aone & \Atwo & \cdots & A_{(k-1)} & [A_{(k)},D] & D\dg
}~. 
}
Thus, $q_\Aki=0$. 
If $\Aone$ is not an interaction term, then $q_\Aki=0$ as well.

For the general case, the same inductive argument as for Proposition~\ref{prop:almost_dp_k} leads to the conclusion.
\end{proof}

Next, we show that the coefficient of $\aki$ is zero 
unless $\Ckm^{\dg} = A_{(k)}$.
\begin{prop}
Assume $Q$ to be a $k$-local conserved quantity of the BLBQ Hamiltonian~\eqref{eq:H_BLBQ} with $J_2=0$ or $J_2=2J_1$.
$q_\aki=0$ holds unless $\ak$ is a doubling operator, that is,
\eq{
\ak = 
\bmatc{
\hline
\Cone & \Cone\dg &          &        &       &          &\\
      & \Ctwo    & \Ctwo\dg &        &       &          &\\
      &          &          & \ddots &       &          &\\
      &          &          &        & \Ckmm & \Ckmm\dg &\\
      &          &          &        &       & \Ckm     &\Ckm\dg \\\hline
}{|ccccccc|}~.
}
\label{prop:dp_singular}
\end{prop}
\begin{proof}
In the case $J_2=0$, the proof is exactly the same as Proposition~\ref{prop:dp_k}. 

We prove for the case $J_2=2J_1$. 
For notational simplicity, we show only for $k=3$, 
but the extension to general $k$ is straightforward as in the main text.
First, we prove for the case of $\Ctwo = \fz$. 
That is, we show that a coefficient of $\Ak = \bmatc{
\hline
\Cone & \Cone \dg& \\
& \fz & F_{m} \\ \hline
}{|ccc|}$ is zero if $ m\neq 0$. 
In fact, we have 
\eq{
q~\bmatc{
\hline
\Cone & \Cone \dg& \\
& \fz & F_{m} \\ \hline
}{|ccc|}_{~i}
\propto
q~\bmatc{
\hline
\fz & F_m & \\
& F_{-m} & F_m \\ \hline
}{|ccc|}_{~i+1}
= 0~,\label{eq:prf_of_doubling_A1}
}
where we applied Proposition~\ref{prop:almost_dp_singular} to 
$\Ak = \bmatc{
\hline
\fz & F_m & \\
& F_{-m} & F_m \\ \hline
}{|ccc|} \propto \fz\ez F_m$
in the right equation. 
Similarly, we can prove for the case of $\Ctwo = F_{\pm2}$ 
by the following calculation: 
\eq{
q~\bmatc{
\hline
\Cone & \Cone \dg& \\
& F_{\pm2} & F_{m} \\ \hline
}{|ccc|}_{~i}
\propto
q~\bmatc{
\hline
F_{\pm2} & F_m & \\
& F_{\pm1-m} & F_{{\mp1+m}} \\ \hline
}{|ccc|}_{~i+1}
= 0~.\label{eq:prf_of_doubling_A2}
}
Finally, we consider the case of $\Ctwo = F_{\pm1}$. 
If $\Athree=F_m$ and $-1\neq\pm m\leq1$, then
we have 
\eq{
q~\bmatc{
\hline
\Cone & \Cone \dg& \\
& F_{\pm1} & F_{m} \\ \hline
}{|ccc|}_{~i}
&\propto
q~\bmatc{
\hline
F_{\pm1} & F_m & \\
& F_{\mp1-m} & F_{{\pm1+m}} \\ \hline
}{|ccc|}_{~i+1}
\nx
&\propto
q~\bmatc{
\hline
\fz & F_{{\pm1+m}} & \\
& F_{\mp1-m} & F_{{\pm1+m}} \\ \hline
}{|ccc|}_{~i+2}
=0~,
}
where we applied Eq.~\eqref{eq:prf_of_doubling_A1}. 
When $\Athree=F_{\pm2}$, then 
\eq{
q~\bmatc{
\hline
\Cone & \Cone \dg& \\
& F_{\pm1} & F_{\pm2} \\ \hline
}{|ccc|}_{~i}
&\propto
q~\bmatc{
\hline
F_{\pm1} & F_{\pm2} & \\
& F_{\mp1} & F_{{\pm1}} \\ \hline
}{|ccc|}_{~i+1}
\nx
&\propto
q~\bmatc{
\hline
F_{\pm2} & F_{{\pm1}} & \\
& F_{\mp1} & F_{{\pm1}} \\ \hline
}{|ccc|}_{~i+2}
=0~,
}
where we applied Eq.~\eqref{eq:prf_of_doubling_A2}. 
This completes the proof.
\end{proof}

Furthermore, 
the coefficients of the doubling operators are proportional to each other, 
as Proposition~\ref{prop:dp_factor_k}.
\begin{prop}
Assume $Q$ to be a $k$-local conserved quantity of the BLBQ Hamiltonian~\eqref{eq:H_BLBQ} with $J_2=0$ or $J_2=2J_1$, 
then the coefficients of all doubling operators are proportional to each other. 
Precisely, their coefficients can be written as
\eq{
&q~\bmatc{
\hline
\Cone & \Cone\dg &          &        &       &          &\\
      & \Ctwo    & \Ctwo\dg &        &       &          &\\
      &          &          & \ddots &       &          &\\
      &          &          &        & \Ckmm & \Ckmm\dg &\\
      &          &          &        &       & \Ckm     &\Ckm\dg \\\hline
}{|ccccccc|}_{~i} \nx
&= q_k c_{(1)} c_{(2)} \cdots c_{(k-2)} c_{(k-1)}~,
\label{eq:dp_factor_singular}
}
by using some constant $q_k$. 
Here, we denote the coupling constant of the interaction term $C_{(n),j} C_{(n),j+1}\dg$ as $c_{(n)}$. 
\label{prop:dp_factor_singular}
\end{prop}
\begin{proof}
Using the fact that for any element of an operator basis 
there is an interaction term that does not commute with it, 
the proof is exactly the same as that of Proposition~\ref{prop:dp_factor_k}. 
\end{proof}

The argument so far holds for $k=2$. 
Therefore, as in the main text, 
we have also shown that there is no $2$-local conserved quantity independent of the Hamiltonian.

\subsection{Step 2 (analysis of $r_{{\bf B}^{k}_i} = 0$)}
The discussion in Step 2 of the main text does not use the fact 
that $f_m$ is not $0$ until Eq.~\eqref{eq:necessary_condition_k}. 
Therefore, Eq.~\eqref{eq:necessary_condition_k} is also correct for $J_2=0$, 
from which the absence of local conserved quantities is shown. 
On the other hand, since $e_m=0$ in the case of $J_2=2J_1$, 
we need to discuss this case separately. 
Specifically, instead of Eq.~\eqref{eq:multiple_outputs}, 
we utilize the following terms:
\eq{
\boldC^k_{i-m} &= \begin{cases}
\bmatc{
\hline
\ez & \ez & \fz &           &  \\
    &     & \fm & \ezo{k-4} & \fp \\\hline
}{|ccccc|}_{~i} \qquad  (m=0) \\ \\
\bmatc{
\hline
\fm & \ezo{m-1} & \fp &     &     &             & \\
    &           & \ez & \ez & \fz &             & \\
    &           &     &     & \fm & \ezo{k-m-4} & \fp \\\hline
}{|ccccccc|}_{~i-m}  \\ \qquad \qquad \qquad \qquad \qquad \qquad \qquad (1\leq m \leq k-4) \\ \\
\bmatc{
\hline
\fm & \ezo{k-4} & \fp &     &     \\
    &           & \ez & \ez & \fz \\\hline
}{|ccccc|}_{~i-k+3}  \quad  (m=k-3)~.
\end{cases}
\label{eq:multiple_outputs_c}
}\linebreak
Together with Eqs.~\eqref{eq:cd_k_left} and~\eqref{eq:cd_k_right} calculated in the main text, 
it is sufficient to obtain the factors of the following commutators:
\eq{
\bmat{
\hline
\Ic{\fp} &\cI{}\\
\Ic{C} & \cI{C^{\dg}}\\
\Ic{} & \cI{\fz}\\ \hline
D & {D^{\dg}} \\ \hline\hline \\[-0.8em] \hline
\Ic{\fp}&  \cI{} \\
\Ic{\ez} &\cI{\ez}\\ \hline
}\qquad
\bmat{
\hline
\Ic{\ez} &\cI{}\\
\Ic{C} & \cI{C^{\dg}}\\
\Ic{} & \cI{\fm}\\ \hline
D & {D^{\dg}} \\ \hline\hline \\[-0.8em] \hline
\Ic{\ez}&  \cI{\fz} \\
\Ic{} &\cI{\fm}\\ \hline
}~,
}
instead of Eqs.~\eqref{eq:cd_middle_left} and~\eqref{eq:cd_middle_right}. 
By a similar argument in the main text, 
we have the following equations of the coefficients: 
\begin{widetext}
\eq{
0 &= r~\bmatc{
\hline
    \ez & \ez & \fz &           &  \\
    &     & \fm & \ezo{k-4} & \fp \\\hline
}{|ccccc|}_{~i} 
= 
+\cfm q~\bmatc{
\hline
\ez & \ez & \fz &           &  \\
    &     & \fm & \ezo{k-5} & \fp \\\hline
}{|ccccc|}_{~i}
\begin{array}{llllll} + q_k f_{-1}^{k-3} 
            (\cfz (&+3/2 \cfp \cfm &+0 \cfpp \cfmm \\
                    &+3/2 \cfm \cfp &+0 \cfmm \cfpp))~,
\end{array}\label{eq:condition_singular_left}}
\eq{
0 &= r~\bmatc{
\hline
\fm & \ezo{m-1} & \fp &     &     &             & \\
    &           & \ez & \ez & \fz &             & \\
    &           &     &     & \fm & \ezo{k-m-4} & \fp \\\hline
}{|ccccccc|}_{~i-m} \nx
&= -\cfm q~\bmatc{
\hline
\Ic{\fm} & \ezo{m-2} & \fp &     &     &             & \cI{}\\
\Ic{}    &           & \ez & \ez & \fz &             & \cI{}\\
\Ic{}    &           &     &     & \fm & \ezo{k-m-4} & \cI{\fp}  \\ \hline}{ccccccc}_{~i-m}\!\!\!\!\!\!\!\!\!\!
+\cfm q~\bmatc{
\hline
\Ic{\fm} & \ezo{m-1} & \fp &     &     &             & \cI{}\\
\Ic{}    &           & \ez & \ez & \fz &             & \cI{}\\
\Ic{}    &           &     &     & \fm & \ezo{k-m-5} & \cI{\fp}  \\ \hline}{ccccccc}_{~i-m+1}
 \nx
&\qquad + q_k f_{-1}^{k-3} \cfz (-3   \cfp \cfm +0 \cfpp \cfmm +9/2 \cfm \cfp +0 \cfmm \cfpp)~,
\label{eq:condition_singular_middle}}
\eq{
0 &= r~\bmatc{
\hline
\fm & \ezo{k-4} & \fp &     &     \\
    &           & \ez & \ez & \fz \\\hline
}{|ccccc|}_{~i-k+3} 
= -\cfm q~\bmatc{
\hline
\fm & \ezo{k-5} & \fp &     &  \\
    &           & \ez & \ez & \fz  \\ \hline}{|ccccc|}_{~i-k+4}
\!\!\!\!\!\!\!\!\!\!
\begin{array}{llllll} + q_k f_{-1}^{k-3} 
            ( \cfz (&-3   \cfp \cfm &+0 \cfpp \cfmm \\
                    &+9/2 \cfm \cfp &+0 \cfmm \cfpp))~.
\end{array}\label{eq:condition_singular_right}}
\end{widetext}
By adding up these equations, we obtain 
\eq{
0 = 3/2q_k (k-1) \cfp f_{-1}^{k-2}\cfz~,
\label{eq:necessary_condition_singular}}
which means $q_k=0$. 
Therefore, we proved that 
the BLBQ model~\eqref{eq:H_BLBQ} with $J_2=2J_1$ has no $k$-local conserved quantity ($3\leq k\leq N/2$). 

\subsection{$k=1,2$ case}
Finally, we consider $k\in\{1,2\}$-local conserved quantities as in the main text. 
As already mentioned at the end of Appendix~\ref{subsec_appx_step1}, 
only the Hamiltonian is a $2$-local conserved quantity. 

For $1$-local conserved quantities, 
we show that, even for $J_2=0,2J_1$, 
the $1$-local conserved quantities are restricted to the total magnetization [Eqs.~\eqref{eq:total_z} and~\eqref{eq:total_pm} for $D=0$]. 
Again, the proof is almost the same as that in the main text. 
For example, we did not use the fact that 
the coupling constant is nonzero in the derivation of Eq.~\eqref{eq:dummy_qf0i}. 
Therefore, $q_{F_0,i}=0$ even in the case of $J_2=0,2J_1$. 
For $J_2=0$, the discussion in the text has already shown that 
the $1$-local conserved quantities are the total magnetization. 
In the case $J_2=2J_1$, only the arguments in Eqs.~\eqref{eq:dummy_qe0i} and~\eqref{eq:dummy_qe+i} need to be changed. 
We consider the following commutator:
\eq{
\bmat{
 &\ez &     \\
 &\fp & \fm \\ \hline
+&\fp & \fm
}\qquad
\bmat{
 &    & \ez \\
 &\fp & \fm \\ \hline
-&\fp & \fm
}~,
}
\eq{
\bmat{
 &\em &     \\
 &\fp & \fm \\ \hline
-&\fz & \fm
}\qquad
\bmat{
 &    & \em \\
 &\fz & \fz \\ \hline
+&\fz & \fm
}~,
}
instead of Eqs.~\eqref{eq:condition_absence_of_E0} and~\eqref{eq:condition_absence_of_E+}. 
These relations read
\eq{
&\cfp q_{E_{0,i}} = \cfp q_{E_{0,i+1}},\nx
&\cfp q_{E_{-1,i}} = \cfz q_{E_{-1,i+1}}
\quad\mbox{for all } i~.
}
The rest of the discussion is the same as in the main text.
In summary, it is shown that 
the $1$-local conserved quantities are restricted to the trivial ones, i.e., 
the total magnetization
even for $J_2=0,2J_1$. 

\begin{widetext}
\section{\change{The anticommutator table for the operator basis}}\label{sec:multiplication_table}
\change{For the convenience of the readers, 
we summarize the anticommutator of two elements in our operator basis in Table~\ref{tbl:AnticommutatorTable}. }

\begin{table*}[h!]
\centering
\caption{ \change{The anticommutator $\{a, b\}$ of elements $a$ and $b$ in the operator basis.}}
  \def\arraystretch{1.5}
\begin{tabular}{|c||c|c|c|c|c|c|c|c|} \hline
   \diagbox[dir=NW]{a}{b} & $\ep$ & $\ez$ & $\em$ & $\fpp$ & $\fp$ & $\fz$ & $\fm$ & $\fmm$ \\ \hline\hline
    $\ep$ & $2\fpp$ & $\fp$ & $\frac{4I-\fz}{3}$ & $0$ &  $0$ & $-\ep$ &$\ez$ &$\em$ \\ \hline
    $\ez$ & $\fp$& $\frac{4I+2\fz}{3}$ & $\fm$ & $0$ & $\ep$ & $2\ez$ & $\em$ & $0$\\ \hline
    $\em$ & $\frac{4I-\fz}{3}$ & $\fm$ & $2\fmm$ & $\ep$ & $\ez$ & $-\em$ & $0$ & $0$ \\ \hline
    $\fpp$ & $0$ & $0$ & $\ep$ & $0$ & $0$ & $2\fpp$ & $-\fp$ & $\frac{2I+\fz}{3}$ \\\hline
    $\fp$ & $0$ & $\ep$ & $\ez$ & $0$ & $-2\fpp$ & $-\fp$ & $\frac{4I-\fz}{3}$ & $-\fm$ \\\hline
    $\fz$ & $-\ep$ & $2\ez$ & $-\em$ & $2\fpp$ & $-\fp$ & $4I-2\fz$ & $-\fm$ & $2\fmm$ \\\hline
    $\fm$ & $\ez$ & $\em$ & $0$ & $-\fp$ & $\frac{4I-\fz}{3}$ & $-\fm$ & $-2\fmm$ & $0$ \\\hline
    $\fmm$ & $\em$ & $0$ & $0$ & $\frac{2I+\fz}{3}$ & $-\fm$ & $2\fmm$ & $0$ & $0$ \\\hline
  \end{tabular}\label{tbl:AnticommutatorTable}
\end{table*}

\section{\change{Necessary conditions for integrability in the extended model}}\label{sec:list_of_eq}
In the extended model~\eqref{eq:generalizedH} defined in Sec.~\ref{sec:necessary}, 
we continue the analysis of Step~2 and 
obtain the following polynomial conditions 
as necessary conditions for the existence of $k$-local conserved quantity for some $k$~\footnote{
In fact, one of the authors (A.~H.) has recently shown that this condition is equivalent to each other for all $3\leq k\leq N/2$~\cite{hokkyo2025}
} satisfying $3\leq k\leq N/2$: 
\eq{3/2\cfz|\cep-\cfp|^2 
+\cez ( -1/6 |\cep|^2 -1/2 \cep \cfm -1/2 \cfp \cem -1/6 |\cfp|^2 + 1/3 |\cfpp|^2)&=0,\label{eq:necessary_integrable_EzEzFz}\\
(\cep-\cfp)(\cfpp(\cez+3\cfz+6 h)-(\cep^2+6\cep\cfp+\cfp^2))-\cfpp^2(\cem-\cfm)&=0,\label{eq:necessary_integrable_CpCpFmm}\\
3\cfz (  \cfmm(\cep - \cfp  )^2 +\cfpp(\cem - \cfm  )^2 ) \qquad \qquad \qquad \qquad \qquad \qquad \qquad \qquad \qquad \nx 
+ |\cfpp|^2 (-4/3\cez^2 +1/3 |\cep|^2 + \cep \cfm + \cfp \cem +1/3 |\cfp|^2  - 1/3 |\cfpp|^2 )&=0,\label{eq:necessary_integrable_FzFppFmm}\\
\cfpp^2(|\cep|^2+3\cep\cfm-3\cfp\cem-|\cfp|^2)\qquad \qquad \qquad \qquad \qquad \qquad \qquad \qquad \qquad \nx 
+\cfpp(\cep+\cfp)(\cep-\cfp)(\cez-21\cfz-6h)
+(\cep+\cfp)(\cep-\cfp)^3&=0,\label{eq:necessary_integrable_EpEzFm}\\
(|\cep-\cfp|^2-2|\cfpp|^2)(\cfpp(\cem-\cfm)^2-\cfmm(\cep-\cfp)^2)+2|\cfpp|^2\cez(\cfp\cem-\cep\cfm)&=0,\label{eq:necessary_integrable_CpFmmCp_1}\\
(11\cez-42\cfz-12h)(\cfp\cem-\cep\cfm)\qquad \qquad \qquad \qquad \qquad \qquad \qquad \qquad \qquad \nx 
+\left(\cfpp(\cem-\cfm)^2-\cfmm(\cep-\cfp)^2\right)
\left(|\cfpp|^2-(\cep-\cfp)(\cem-\cfm)\right)&=0,\label{eq:necessary_integrable_CpFmmCp_2}\\
(\cep-\cfp)\left(9/2\cfz\left(\cfmm(\cep-\cfp)^2-\cfpp(\cem-\cfm)^2\right)
-|\cfpp|^2\left(1/2|\cfpp|^2+\cez(\cez-3\cfz-6h)\right)
\right)\nx
+(\cem-\cfm)|\cfpp|^2(-3\cez\cfpp+1/2(\cep^2+6\cep\cfp+\cfp^2))
&=0,\label{eq:necessary_integrable_CpDmEz_1}\\
(\cep-\cfp)\Bigl(3/2\cfz\left(\cfmm(\cep-\cfp)^2-\cfpp(\cem-\cfm)^2\right) \qquad \qquad \qquad \qquad  \qquad \qquad \qquad \qquad\qquad \qquad \nx
+|\cfpp|^2\Bigl(1/6|\cfpp|^2 -1/6|\cep|^2 -1/6|\cfp|^2 +3/2\cep\cfm +3/2\cem\cfp
+\cez^{2}/3 -18\cfz^2 -3\cez\cfz +18\cfz h\Bigr)\Bigr)\nx
+(\cem-\cfm)|\cfpp|^2(3\cfz\cfpp -4/3\cep\cfp)
&=0,\label{eq:necessary_integrable_CpCmEz_1}\\
(\cep+\cfp)\Bigl((2\cez+9/2\cfz)\cfmm(\cep-\cfp)^2-9/2\cfz\cfpp(\cem-\cfm)^2\qquad \qquad \qquad\qquad \qquad \nx
+|\cfpp|^2 \bigl(
-|\cfpp|^{2}/2  +2\cez^2 +|\cep-\cfp|^{2}/2 
\bigr)\Bigr)
&=0,\label{eq:necessary_integrable_CpDmEz_2}\\
(\cep+\cfp)\left(3/2\cfz(5\cfmm(\cep-\cfp)^2-\cfpp(\cem-\cfm)^2)+|\cfpp|^2(1/6|\cfpp|^2 +1/3\cez^2+6\cez\cfz-27\cfz^2)\right)\nx
+1/6(\cep-\cfp)|\cfpp|^2(|\cfp|^2 -|\cep|^2 +9\cep\cfm -9\cem\cfp)
&=0.\label{eq:necessary_integrable_CpCmEz_2}%
}

\end{widetext}

Here, we summarize our results for the extended model~\eqref{eq:generalizedH}. 
\begin{thm}\label{thm:extended_2}
In the extended model~\eqref{eq:generalizedH} with nonzero coupling constants,  
$k$-local conserved quantities with $3 \leq k \leq N/2$ are absent 
unless  all of Eqs.~\eqref{eq:necessary_integrable_EzEzFz}-\eqref{eq:necessary_integrable_CpCmEz_2}
hold.
\end{thm}

Theorem~\ref{thm:extended_2} represents \textit{necessary} conditions for integrability. 
Therefore, these equations provide a simple way to determine whether the model of interest is integrable or not, 
complementary to solving the Yang--Baxter equation
which gives a \textit{sufficient} condition for integrability. 
In particular, the above equation shows the necessity of fine-tuning of parameters for integrable systems. 
In other words, an integrable system becomes nonintegrable 
when perturbed in such a way that either of Eqs.~\eqref{eq:necessary_integrable_EzEzFz}-\eqref{eq:necessary_integrable_CpCmEz_2}
is no longer satisfied. 
For example, given an integrable system with $\cep\neq\cfp$, 
we can see from Eq.~\eqref{eq:necessary_integrable_CpCpFmm} that changing the anisotropic field $h$ breaks the integrability. 

Assuming that the coupling constants $e_m$ and $f_m$ are real, 
these polynomial relations can be simplified. 
For example, Eqs.~\eqref{eq:necessary_integrable_CpFmmCp_1} and \eqref{eq:necessary_integrable_CpFmmCp_2} 
become trivial identities. 
The remaining equations take the following form:
\begin{widetext}
\eq{3/2\cfz(\cep-\cfp)^2 
+\cez ( -1/6 \cep^2 -\cep \cfp-1/6 \cfp^2 + 1/3\cfpp^2)
&=0,\tag{\ref{eq:necessary_integrable_EzEzFz}$'$}\\
(\cep-\cfp)(\cfpp(\cez+3\cfz+6 h)-(\cep^2+6\cep\cfp+\cfp^2)-\cfpp^2)
&=0,\tag{\ref{eq:necessary_integrable_CpCpFmm}$'$}\\
6\cfz(\cep - \cfp  )^2 
+ \cfpp (-4/3\cez^2 +1/3 \cep^2 + 2\cep \cfp +1/3 \cfp^2  - 1/3 \cfpp^2 )
&=0,\tag{\ref{eq:necessary_integrable_FzFppFmm}$'$}\\
(\cep+\cfp)(\cep-\cfp)(\cfpp^2+\cfpp(\cez-21\cfz-6h)+(\cep-\cfp)^2)
&=0,\tag{\ref{eq:necessary_integrable_EpEzFm}$'$}\\
(\cep-\cfp)\left(
-1/2\cfpp^2-\cez(\cez-3\cfz-6h)
-3\cez\cfpp+1/2(\cep^2+6\cep\cfp+\cfp^2)
\right)
&=0,\tag{\ref{eq:necessary_integrable_CpDmEz_1}$'$}\\
(\cep-\cfp)\Bigl(
1/6\cfpp^2 -1/6\cep^2 -1/6\cfp^2 +5/3\cep\cfp
+\cez^{2}/3 -18\cfz^2 -3\cez\cfz +18\cfz h+3\cfz\cfpp \Bigr)
&=0,\tag{\ref{eq:necessary_integrable_CpCmEz_1}$'$}\\
(\cep+\cfp)\Bigl((2\cez+1/2\cfpp)(\cep-\cfp)^2
+\cfpp \bigl(-\cfpp^2/2  +2\cez^2 \bigr)\Bigr)
&=0,\tag{\ref{eq:necessary_integrable_CpDmEz_2}$'$}\\
(\cep+\cfp)\left((6\cfz-1/6\cfpp)(\cep-\cfp)^2+\cfpp(1/6\cfpp^2 +1/3\cez^2+6\cez\cfz-27\cfz^2)\right)
&=0.\tag{\ref{eq:necessary_integrable_CpCmEz_2}$'$}%
}
\end{widetext}
By solving these equations, we obtain the six classes of solutions listed in Theorem~\ref{thm:extended}.

}%

\bibliography{document}
\end{document}